\documentclass[12pt]{article}
\pdfoutput=1 
\usepackage[utf8]{inputenc}
\usepackage{amsmath,amssymb,amsfonts}
\usepackage{float,graphicx, multirow,setspace}
\usepackage[caption=false]{subfig}
\usepackage[margin=1in]{geometry}
\usepackage{url}
\usepackage{enumitem}
\usepackage{stackengine}
\usepackage{dsfont,bbm}
\usepackage{verbatim}
\usepackage[scr=boondox]{mathalfa}
\usepackage{subfiles}
\usepackage{xr}
\usepackage{array,multirow}
\usepackage[dvipsnames]{xcolor}
\usepackage{setspace}
\usepackage{algorithm2e}
\usepackage[margin=1in]{geometry}
\usepackage{hyperref}
\usepackage{forest}
\usepackage[section]{placeins}

\newcommand{\E}{{\mathbb E}}

\renewcommand{\o}{{\mathscr o}}

\renewcommand{\P}{{\mathbb P}}
\newcommand{\X}{{\mathscr X}}
\newcommand{\C}{{\mathscr{C}}}
\newcommand{\T}{{\mathscr{T}}}
\newcommand{\R}{{\mathscr{R}}}
\newcommand{\OO}{{\mathscr O}}
\newcommand{\bP}{{ \mathbf{P} }}
\newcommand{\bQ}{{ \mathbf{Q} }}
\newcommand{\bR}{{ \mathbf{R} }}
\newcommand{\bm}{{ \mathbf{m} }}

\newcommand{\Wopp}{{ W_{\text{opp}} }}
\newcommand{\B}{{\mathscr B}}
\newcommand{\one}{\mathbbm{1}}

\renewcommand{\R}{{\mathscr{R}}}
\newcommand{\rd}{{\mathsf{rd} }}
\renewcommand{\dbinom}{ \mathsf{dbinom} }
\newcommand{\Binom}{ \mathsf{Binom} }

\newcommand{\Rstu}{{\mathsf{R}}}
\newcommand{\Pchalky}{ \bR^{(\mathsf{cc})} }

\usepackage[compact]{titlesec}
\titlespacing{\section}{0pt}{1ex}{1ex}
\titlespacing{\subsection}{0pt}{1ex}{0ex}
\titlespacing{\subsubsection}{0pt}{0.5ex}{0ex}

\usepackage{amsthm}
\DeclareMathOperator*{\argmax}{arg\,max}

\usepackage{fullpage}
\usepackage{parskip}
\usepackage{natbib}
\RestyleAlgo{ruled}
\newtheorem{theorem}{Theorem}
\numberwithin{equation}{section}
\definecolor{SkyBlue}{RGB}{14, 118, 188}
\definecolor{BrightRed}{RGB}{223,82, 78}
\hypersetup{pdfborder = {0 0 0.5 [3 3]}, colorlinks = true, linkcolor = BrightRed, citecolor = SkyBlue, filecolor = BrightRed}

\title{Entropy-Based Strategies for Multi-Bracket Pools}
\author{Ryan S. Brill\thanks{Graduate Group in Applied Mathematics and Computational Science, University of Pennsylvania. Correspondence to: ryguy123@sas.upenn.edu}, \ Abraham J. Wyner\thanks{Department of Statistics and Data Science, The Wharton School, University of Pennsylvania}, \ and Ian J. Barnett\thanks{Department of Biostatistics, Perelman School, University of Pennsylvania}}

\begin{document}
\maketitle

\begin{abstract}
Much work in the parimutuel betting literature has discussed estimating event outcome probabilities or developing optimal wagering strategies, particularly for horse race betting.  Some betting pools, however, involve betting not just on a single event, but on a tuple of events.  For example, pick six betting in horse racing, March Madness bracket challenges, and predicting a randomly drawn bitstring each involve making a series of individual forecasts.  Although traditional optimal wagering strategies work well when the size of the tuple is very small (e.g., betting on the winner of a horse race), they are intractable for more general betting pools in higher dimensions (e.g., March Madness bracket challenges). Hence we pose the multi-brackets problem: supposing we wish to predict a tuple of events and that we know the true probabilities of each potential outcome of each event, what is the best way to tractably generate a set of $n$ predicted tuples?  The most general version of this problem is extremely difficult, so we begin with a simpler setting.  In particular, we generate $n$ independent predicted tuples according to a distribution having optimal entropy.  This entropy-based approach is tractable, scalable, and performs well. 


\end{abstract}




\section{Introduction}

\textit{Parimutuel betting} or \textit{pool betting} involves pooling together all bets of a particular type on a given event, deducting a track take or vigorish, and splitting the pot among all winning bets.
Prime examples are horse race betting and the March Madness bracket challenge (which involves predicting the winner of each game in the NCAA Division I men's basketball ``March Madness'' tournament).
Profitable parimutuel wagering systems have two components: a probability model of the event outcome and a bet allocation strategy.
The latter uses the outcome probabilities as inputs to a betting algorithm that determines the amount to wager on each potential outcome.
There is a large body of literature on estimating outcome probabilities for pool betting events.
For instance, we provide an overview of estimating outcome probabilities for horse races and college basketball matchups in Appendix~\ref{app:probEstimationDetails}.
There is also a large body of literature on developing optimal wagering strategies, particularly for betting on horse race outcomes.
Notably, assuming the outcome probabilities are known, \citet{Isaacs1965} and \citet{kelly1956} derive the amount to wager on each horse so as to maximize expected profit and expected log wealth, respectively.
\citet{rosner1975} derives a wagering strategy for a risk averse decision maker, and
\citet{Willis1964OptimumNS} and \citet{hausch1981} derive other wagering strategies.
On the other hand, there has been very limited work, and no literature to our knowledge, on deriving optimal strateges for generating multiple predicted March Madness brackets.
Existing work focuses on generating a single predicted bracket (see Appendix~\ref{app:prevWorkMMmultipleBrackets} for details).

Existing wagering strategies for pools that involve betting on the outcome of a single event (e.g., the winner of a horse race) have been successful.
For instance, \citet{benter} reported that his horse race gambling syndicate made \textit{significant} profits during its five year gambling operation.   
However, many betting pools in the real world involve betting not just on a single event, but on a \textit{tuple of events}.
For example, the pick six bet in horse racing involves predicting the winner of each of six horse races.
Also, the March Madness bracket challenge involves predicting the winner of each game in the NCAA Division I men's basketball tournament.
Another compelling example to the pure mathematician is predicting each of the bits in a randomly drawn bitstring.
In each of these three prediction contests, the goal is to predict as best as possible a tuple of events, which we call a \textit{bracket}.
We suppose it is permissible to generate multiple predicted brackets, so we call these contests \textit{multi-bracket pools}.
In developing wagering strategies for multi-bracket pools, the literature on estimating outcome probabilities for each event in the bracket still applies.
However, given these probabilities, the wagering strategy literature developed for betting on single events doesn't extend to general multi-bracket pools.
Although these methods work well in low dimensional examples such as betting on the winner of a horse race, they are intractable for general multi-bracket pools having larger dimension (e.g., March Madness bracket challenges); extensions of classical analytical solutions are exponential in the size of a bracket.

Hence we pose the \textit{multi-brackets problem}.
Suppose we wish to predict a bracket (a tuple of events) and suppose we know the true probabilities of each potential outcome of each event.
Then, what is the best way to tractably generate a set of $n$ predicted brackets?
More concretely, how can we construct a set of $n$ brackets that maximize an objective function such as expected score, win probability, or expected profit?
The most general version of the multi-brackets problem, which finds the optimal set of $n$ brackets across all such possible sets, is extremely difficult.
To make the problem tractable, possible, and/or able to be visualized, depending on the particular specification of the multi-bracket pool, we make simpifying assumptions.
First, we assume we (and optionally a field of opponents) predict i.i.d. brackets generated according to a bracket distribution.
The task becomes to find the optimal generating bracket distribution.
For higher dimensional examples (e.g., March Madness bracket challenges), we make another simplifying assumption, optimizing over a smartly chosen low dimensional subspace of generating bracket distributions.
In particular, we optimize over brackets of varying levels of entropy.
We find that this entropy-based approach is sufficient to generate well-performing sets of bracket predictions.
We also learn the following high-level lessons from this strategy: we should increase the entropy of our bracket predictions as $n$ increases and as our opponents increase entropy.

The remainder of this paper is organized as follows.
In Section~\ref{sec:formulateProblem} we formally introduce the multi-brackets problem.
Then in Section~\ref{sec:guessBitstring} we propose an entropy-based solution to what we consider a canonical example of a multi-bracket pool: guessing a randomly drawn bitstring.
Using this canonical example to aid our understanding of multi-bracket pools, in Section~\ref{sec:ep_sec} we make the connection between the multi-brackets problem and information theory, particularly through the Asymptotic Equipartition Property.
Then, in Section~\ref{sec:realWorldExamples} we propose entropy-based solutions to real world examples of multi-bracket pools, including the pick six bet in horse racing in Section~\ref{sec:ex_pickSix} and March Madness bracket challenges in Section~\ref{sec:ex_MM}.
We conclude in Section~\ref{sec:discussion}.

\section{The multi-brackets problem}\label{sec:formulateProblem}

In this section we formally introduce the multi-brackets problem.
The goal of a multi-bracket pool is to predict a tuple of $m$ outcomes $\tau = (\tau_1,...,\tau_m)$, which we call the ``true'' observed reference \textit{bracket}.
We judge how ``close'' a bracket prediction $x = (x_1,...,x_m)$ is to $\tau$ by a \textit{bracket scoring function} $f(x,\tau)$.
One natural form for the scoring function is
\begin{equation}
    f(x,\tau) = \sum_{i=1}^{m} w_i \cdot \one\{x_i = \tau_i\},
\label{eqn:bracketScoringFunctionWeighted}
\end{equation}
which is the number of outcomes predicted correctly weighted by $\{w_i\}_{i=1}^{m}$.
Another is
\begin{equation}
    f(x,\tau) = \one\{x = \tau\},
\label{eqn:bracketScoringFunctionExact}
\end{equation}
which is one if and only if the predicted bracket is exactly correct.
The contestants who submit the highest scoring brackets win the pool.

The multi-brackets problem asks the general question: if we could submit $n$ brackets to the pool, how should we choose which brackets to submit?
This question takes on various forms depending on the information available to us and the structure of a particular multi-bracket pool.
In the absence of information about opponents' predicted brackets, 
how should we craft our submitted set $\B_n$ of $n$ bracket predictions in order to maximize expected maximum score?
Formally, solve
\begin{align}
    \B_n^* :=& \argmax_{\{\B_n \subset \X: |B_n| = n\}} \E_{\tau} \bigg[ \max_{x \in \B_n} f(x,\tau) \bigg].
    \label{eqn:maximize_expected_score}
\end{align}
Or, assuming a field of opponents submits a set $\OO_k$ of $k$ bracket predictions to the pool according to some strategy, how should we craft our submitted set $\B_n$ of $n$ brackets in order to maximize our probability of having the best bracket?
Formally, solve
\begin{align}
    \B_n^* :=& \argmax_{\{\B_n \subset \X: |B_n| = n\}} \P_{\tau, \OO_k} \bigg[ \max_{x \in \B_n} f(x,\tau) \geq \max_{y \in \OO_k} f(y,\tau) \bigg].
    \label{eqn:maximize_win_prob}
\end{align}
Another version of a multi-bracket pool offers a \textit{carryover} $C$ of initial money in the pot, charges $b$ dollars per submitted bracket, and removes a fraction $\alpha$ from the pot as a track take or vigorish.
The total pool of money entered into the pot is thus
\begin{align}
    T = C + b(n+k)(1-\alpha),
    \label{eqn:totalMoney}
\end{align}
which is split among the entrants with the highest scoring brackets.
The question becomes: how should we craft our submitted set $\B_n$ of $n$ brackets in order to maximize expected profit?
Formally, solve
\begin{align}
    \B_n^* :=& \argmax_{\{\B_n \subset \X: |B_n| = n\}} T \cdot \P_{\tau, \OO_k} \bigg[ \max_{x \in \B_n} f(x,\tau) > \max_{y \in \OO_k} f(y,\tau) \bigg] - b\cdot n.
    \label{eqn:maximize_exp_profit}
\end{align}
This variant assumes no ties but is easily extended to incorporate ties (see Section~\ref{sec:ex_pickSix}).
The optimization problems in Equations \eqref{eqn:maximize_expected_score}, \eqref{eqn:maximize_win_prob}, and \eqref{eqn:maximize_exp_profit} and related variants define the multi-brackets problem. 

In upcoming sections we explore specific examples of the multi-brackets problem.
In guessing a randomly drawn bitstring (Section~\ref{sec:guessBitstring}) and the March Madness bracket challenge (Section~\ref{sec:ex_MM}) we explore the multi-brackets problem via scoring function \ref{eqn:bracketScoringFunctionWeighted} and objective functions \ref{eqn:maximize_expected_score} and \ref{eqn:maximize_win_prob}.
In pick six betting in horse racing (Section~\ref{sec:ex_pickSix}) we explore the multi-brackets problem via scoring function \ref{eqn:bracketScoringFunctionExact} and objective function \ref{eqn:maximize_exp_profit}.

The most general version of the multi-brackets problem, which finds the optimal set of $n$ brackets across all such possible sets, is extremely difficult.
To make the problem tractable, possible, and/or able to be visualized, depending on the particular specification of the multi-bracket pool, we make simpifying assumptions.
We assume we (and the field of opponents) submit i.i.d. brackets generated from some bracket distribution.
As the size of a bracket increases, solving the multi-brackets problem under this assumptions quickly becomes intractable, so we optimize over smartly chosen low dimensional subspaces of bracket distributions.
We find this entropy-based strategy is sufficient to generate well-performing sets of submitted brackets.

\section{Canonical example: guessing a randomly drawn bitstring}\label{sec:guessBitstring}

In this section we delve into what we consider a canonical example of a multi-bracket pool: guessing a randomly drawn bitstring.
In this contest, we want to predict the sequence of bits in a reference bitstring, which we assume is generated according to some known probability distribution. 
We submit $n$ guesses of the reference bitstring with the goal of being as ``close'' to it as possible or of being ``closer'' to it than a field of $k$ opponents' guesses, according to some distance function.
With some assumptions on the distribution $\bP$ from which the reference bitstring is generated, the distribution $\bQ$ from which we generate bitstring guesses, and the distribution $\bR$ from which opponents generate bitstring guesses, expected maximum score and win probability are analytically computable and tractable.
By visualizing these formulas we discern high-level lessons relevant to all multi-bracket pools.
To maximize the expected maximum score of a set of $n$ submitted randomly drawn brackets, we should increase the entropy of our submitted brackets as $n$ increases.
To maximize the probability that the maximum score of $n$ submitted randomly drawn brackes exceeds that of $k$ opposing brackets, we should increase the entropy of our brackets as our opponents increase entropy.

The objective of this multi-bracket pool is to predict a randomly drawn bitstring, which is to predict a sequence of bits.
Here, a bracket is a bitstring 
consisting of $m$ bits divided into $R$ rounds with $m_{\rd}$ bits in each round $\rd \in \{1,...,R\}$.
For concreteness we let there be $m_{\rd} = 2^{R - \rd}$ bits in each of $R = 6$ rounds (i.e., $32$ bits in round $1$, $16$ bits in round $2$, $8$ bits in round $3$, ..., $1$ bit in round $6$, totaling $63$ bits), but the analysis in this Section holds for other choices of $m_{\rd}$ and $R$.
The ``true'' reference bracket that we are trying to predict is a bitstring $\tau = (\tau_{\rd,i} : 1\leq \rd\leq R, 1 \leq i \leq m_\rd)$.
A field of opponents submits $k$ guesses of $\tau$, the brackets $(y^{(1)},...,y^{(k)})$, where each bracket is a bitstring $y^{(\ell)} = (y^{(\ell)}_{\rd,i} : 1\leq \rd\leq R, 1 \leq i \leq m_\rd)$. 
We submit $n$ guesses of $\tau$, the brackets $(x^{(1)},...,x^{(n)})$, where each bracket is a bitstring $x^{(j)} = (x^{(j)}_{\rd,i} : 1\leq \rd\leq R, 1 \leq i \leq m_\rd)$.
The winning submitted bracket among $\{x^{(j)}\}_{j=1}^{n} \cup \{y^{(\ell)}\}_{\ell=1}^{k}$ is ``closest'' to the reference bracket $\tau$ according to a scoring function $f(x,\tau)$ measuring how ``close'' $x$ is to $\tau$.
Here, we consider 
\begin{equation}
    f(x,\tau) = \sum_{\rd=1}^{R} \sum_{i=1}^{m_{\rd}} w_{\rd,i} \cdot \one\{x_{\rd,i} = \tau_{\rd,i}\},
    \label{eqn:espn_score}
\end{equation}
which is the weighted number of bits guessed correctly.
This scoring function encompasses both \textit{Hamming score} and \textit{ESPN score}.
Hamming score measures the number of bits guessed correctly, weighing each bit equally ($w_{\rd,i} \equiv 1$). 
ESPN score weighs each bit by $w_{\rd,i} = 10\cdot 2^{\rd-1}$ so that the maximum accruable score in each round is the same ($10\cdot2^{R-1}$).

Suppose the true reference bitstring $\tau$ is generated according to some known distribution $\bP$ and opponents' bitstrings are generated according to some known distribution $\bR$.
Our task is to submit $n$ predicted bitstrings so as to maximize expected maximum score
\begin{align}
    \E\bigg[\max_{j=1,...,n} f(x^{(j)},\tau)\bigg]
    \label{eqn:SMM_eMaxScore}
\end{align}
or the probability that we don't lose the bracket challenge
\begin{align}
    \P\bigg[ \max_{j=1,...,n} f(x^{(j)},\tau) \geq \max_{\ell=1,...,k} f(y^{(\ell)},\tau) \bigg].
    \label{eqn:SMM_wp}
\end{align}
In particular, we wish to submit $n$ bitstrings generated according to some distribution $\bQ$, and it is our task to find suitable $\bQ$.
For tractability, we consider the special case that bits are drawn independently with probabilities varying by round.
We suppose that each bit $\tau_{\rd,i}$ in the reference bitstring is an independently drawn $\text{Bernoulli}(p_\rd)$ coin flip.
The parameter $p_\rd \in [0.5, 1]$ controls the entropy of the contest: lower values correspond to a higher entropy (more variable) reference bitstring that is harder to predict.
By symmetry, our strategy just needs to vary by round.
So, we assume that each of our submitted bits $x^{(j)}_{\rd,i}$ is an independently drawn $\text{Bernoulli}(q_\rd)$ coin flip and each of our opponents' submitted bits $y^{(\ell)}_{\rd,i}$ is an independently drawn $\text{Bernoulli}(r_\rd)$ coin flip.
The parameters $q_\rd$ and $r_\rd$ control the entropy of our submitted bitstrings and our opponents' submitted bitstrings, respectively.
Our task is to find the optimal strategy or entropy level $(q_\rd)_{\rd=1}^{R}$.
In this setting, expected maximum score and win probability are analytically computable and tractable (see Appendix~\ref{app:smm_details}). 

We first visualize the case where the entropy of the reference bitstring, our submitted bitstrings, and our opponents' submitted bitstrings don't vary by round: $p \equiv p_\rd$, $q \equiv q_\rd$, and $r \equiv r_\rd$.
In Figure~\ref{fig:plot_smm_eMaxScore} we visualize the expected maximum Hamming score of $n$ submitted bitstrings as a function of $p$, $q$, and $n$.
We find that we should increase the entropy of our submitted brackets (decrease $q$) as $n$ increases, transitioning from pure ``chalk'' ($q=1$) for $n=1$ bracket to the true amount of randomness ($q=p$) for large $n$.
Specifically, for small $n$ the green line $q = p$ lies below the blue lines (large $q$), and for large $n$ the green line lies above all the other lines.

\begin{figure}[htb!]
    \centering
    \includegraphics[width=17cm]{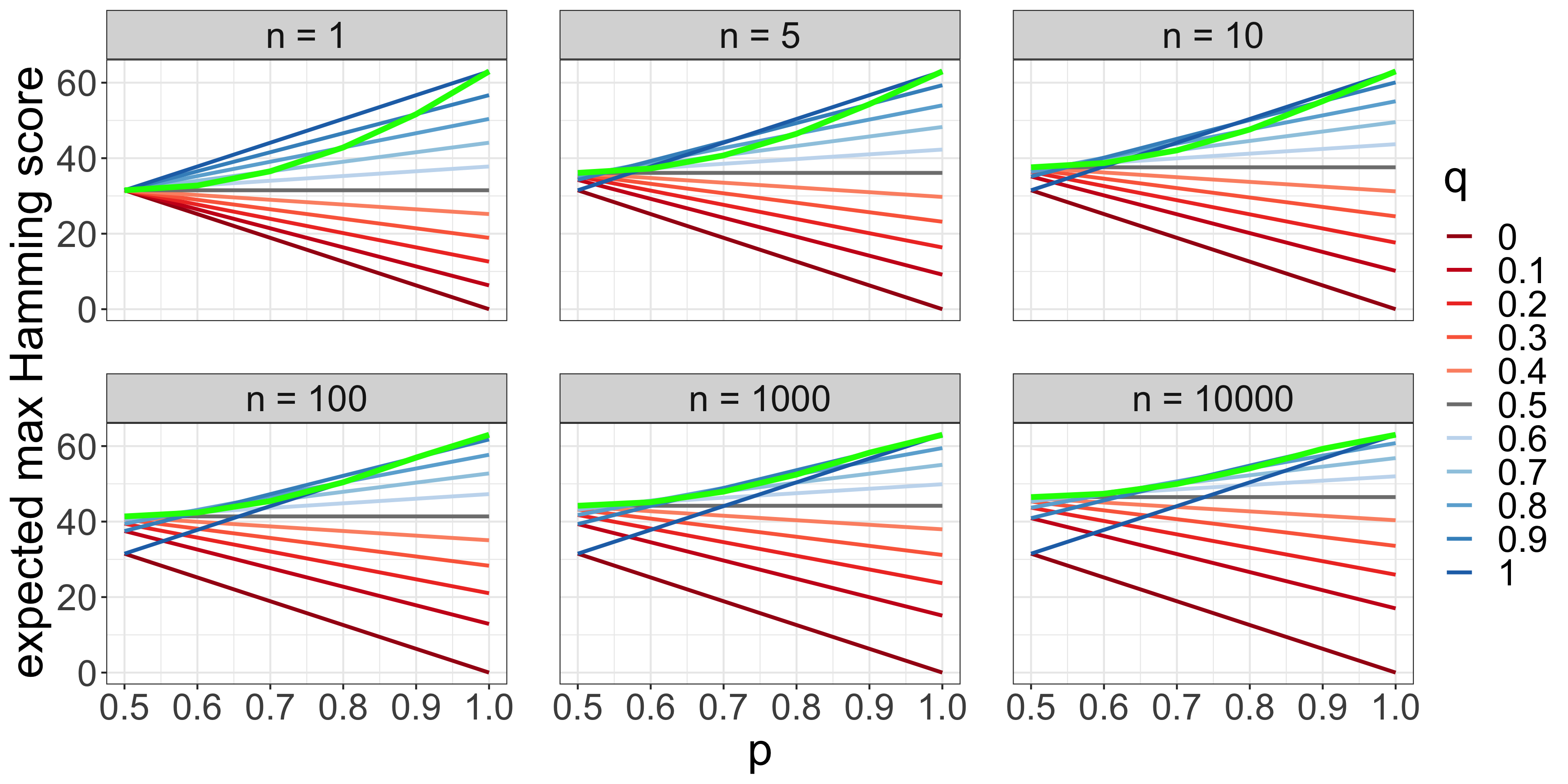}
    \caption{
        The expected maximum Hamming score ($y$-axis) of $n$ submitted Bernoulli($q$) bitstrings relative to a reference Bernoulli($p$) bitstring as a function of $p$ ($x$-axis), $q$ (color), and $n$ (facet) in the ``guessing a randomly drawn bitstring'' contest with $p \equiv p_\rd$, $q \equiv q_\rd$, $r \equiv r_\rd$, and $R=6$ rounds.
        As $n$ increases, we want to increase the entropy of our submitted brackets.
    }
    \label{fig:plot_smm_eMaxScore}
\end{figure}

In Figure~\ref{fig:plot_wpHammingScore_k100} we visualize win probability as a function of $q$, $r$, and $n$ for $k=100$ and $p=0.75$.
The horizontal gray dashed line $q=p=0.75$ represents that we match the entropy of the reference bitstring, the vertical gray dashed line $r=p=0.75$ represents that our opponents match the entropy of the reference bitstring, and the diagonal gray dashed line $q=r$ represents that we match our opponents' entropy.
We should increase entropy (decrease $q$) as $n$ increases, visualized by the green region moving downwards as $n$ increases.
Further, to maximize win probability, we should increase entropy (decrease $q$) as our opponents' entropy increases (as $r$ decreases), visualized by the triangular form of the green region.
In other words, we should tailor the entropy of our brackets to the entropy of our opponents' brackets.
These trends are similar for other values of $k$ and $n$ (see Figure~\ref{fig:wp_smm_varyK} of Appendix~\ref{app:smm_details}).

\begin{figure}[htb!]
    \centering
    \includegraphics[width=\textwidth]{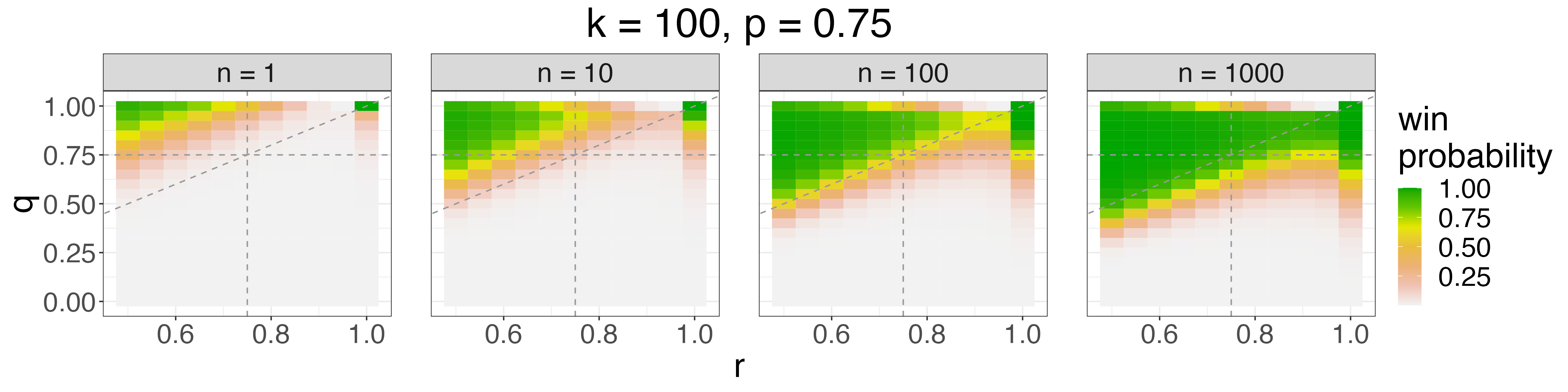}
    \caption{
        The probability (color) that the maximum Hamming score of $n$ submitted Bernoulli($q$) brackets relative to a reference Bernoulli($p$) bracket exceeds that of $k$ opposing Bernoulli($r$) brackets as a function of $q$ ($y$-axis), $r$ ($x$-axis), and $n$ (facet) for $p=0.75$ and $k=100$ in the ``guessing a randomly drawn bitstring'' contest with $p \equiv p_\rd$, $q \equiv q_\rd$, $r \equiv r_\rd$, and $R=6$ rounds.
        We should increase entropy as $n$ increases and as our opponents' entropy increases.
    }
    \label{fig:plot_wpHammingScore_k100}
\end{figure}

These trends generalize to the case where the entropy of each bitstring varies by round (i.e., general $q_\rd$, $r_\rd$, and $p_\rd$).
It is difficult to visualize the entire $R=6$ dimensional space of $p=(p_1,...,p_6)$, $q=(q_1,...,q_6)$, and $r=(r_1,...,r_6)$ so we instead consider a lower dimensional subspace.
Specifically we visualize a 2 dimensional subspace of $q$ parameterized by $(q_E,q_L)$, where $q_E$ denotes $q$ in early rounds and $q_L$ denotes $q$ in later rounds.
For example, $q_E = q_1 = q_2 = q_3$ and $q_L = q_4 = q_5 = q_6$ is one of the five possible partitions of $(q_E,q_L)$.
We similarly visualize a 2 dimensional subspace of $r$ parameterized by $(r_E,r_L)$.
Finally, we let the reference bitstring have a constant entropy across each round, $p_\rd \equiv p$.

In Figure~\ref{fig:plot_smm_eMaxScoreRd} we visualize the expected maximum ESPN score of $n$ bitstrings as a function of $q_E$, $q_L$, and $n$ for $p=0.75$.
The three columns display the results for $n=1$, $n=10$, and $n=100$, respectively.
The five rows display the results for the five partitions of $(q_E,q_L)$.
For instance, the first row shows one partition $q_E = q_1$ and $q_L = q_2 = q_3 = q_4 = q_5 = q_6$.
As $n$ increases, the expected maximum ESPN score increases.
We visualize this as the lines moving upwards as we move right across the grid of plots.
As $E$ increases (i.e., as $q_E$ encompasses a larger number of early rounds), the impactfulness of the late round strategy $q_L$ decreases.
We visualize this as the lines becoming more clumped together as we move down the grid of plots in Figure~\ref{fig:plot_smm_eMaxScoreRd}.
For $n=1$, the best strategy is pure chalk ($q_E=1$, $q_L=1$), and as $n$ increases, the optimal values of $q_E$ and $q_L$ decrease.
In other words, as before, we want to increase the entropy of our submitted brackets as $n$ increases.
We visualize this as the circle (i.e., the best strategy in each plot) moving leftward and having a more reddish color as $n$ increases.

\begin{figure}[htb!]
    \centering
    \includegraphics[width=17cm]{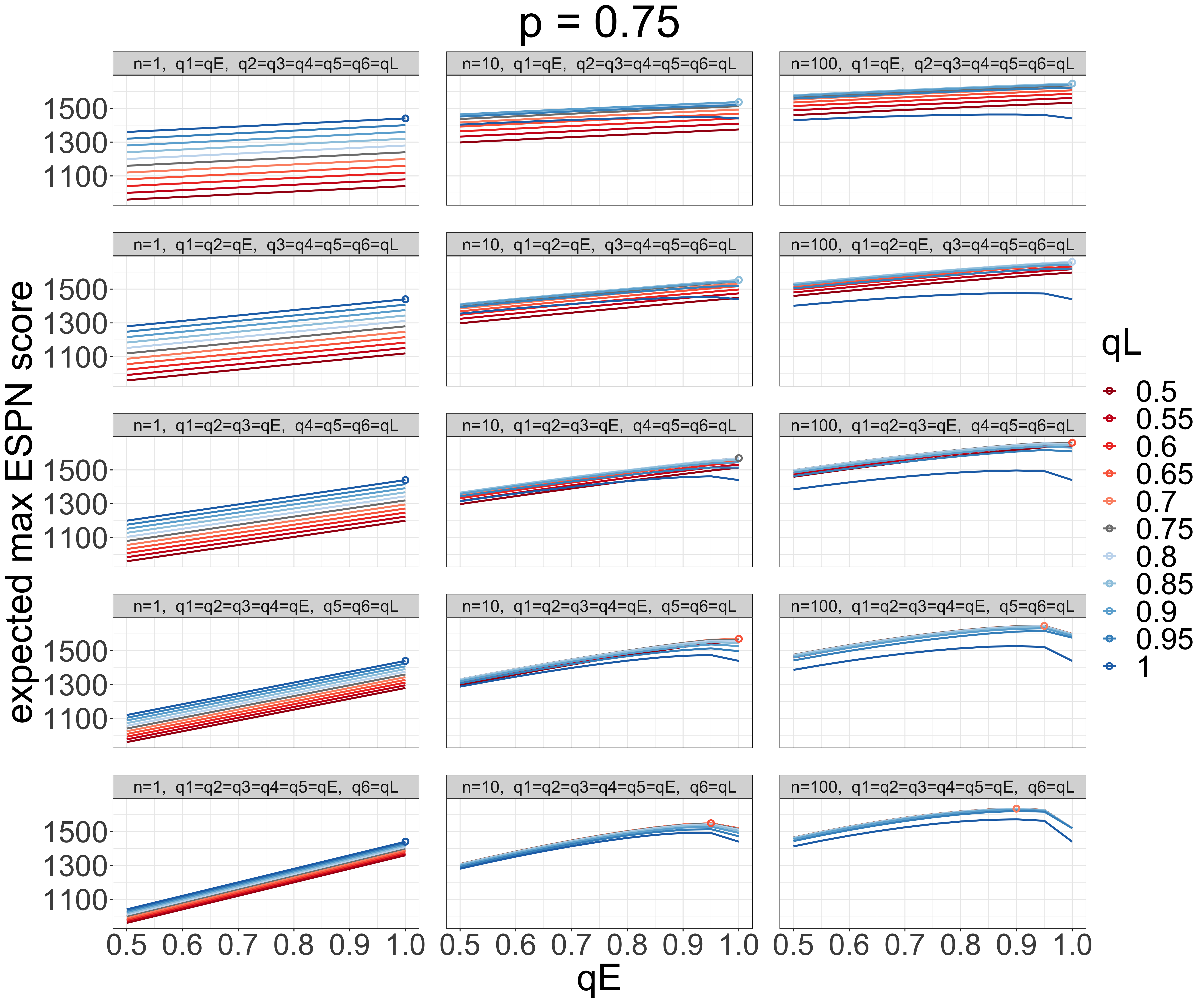}
    \caption{
        The expected maximum ESPN score ($y$-axis) of $n$ submitted bitstrings, with Bernoulli($q_E$) bits in early rounds and Bernoulli($q_L$) bits in later rounds, relative to a reference Bernoulli($p$) bitstring as a function of $q_E$ ($x$-axis), $q_L$ (color), $n$ (columns), and the partition $(q_E, q_L)$ (rows) in the ``guessing a randomly drawn bitstring'' contest with $R=6$ rounds and $p=0.75$.
        The circles indicates the best strategy in each setting.
        As $n$ increases, we want to increase the entropy of our bracket predictions in both early and late rounds.
    }
    \label{fig:plot_smm_eMaxScoreRd}
\end{figure}

\begin{figure}[hbt!]
\centering
\subfloat[]{
    \includegraphics[width = .5\textwidth]{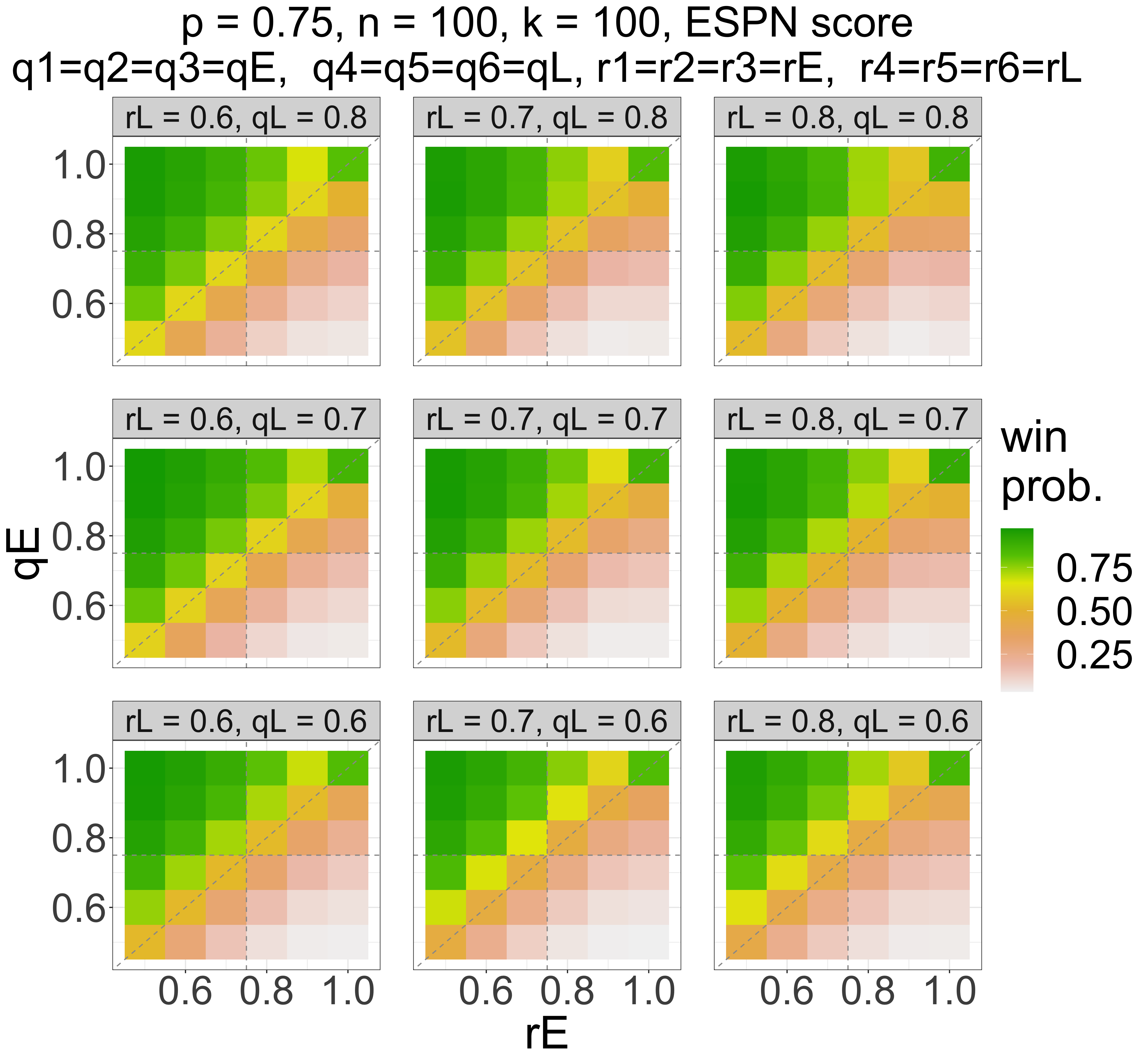}
    \label{fig:plot_smm_rd_wp_1}
} 
\subfloat[]{
    \includegraphics[width = .5\textwidth]{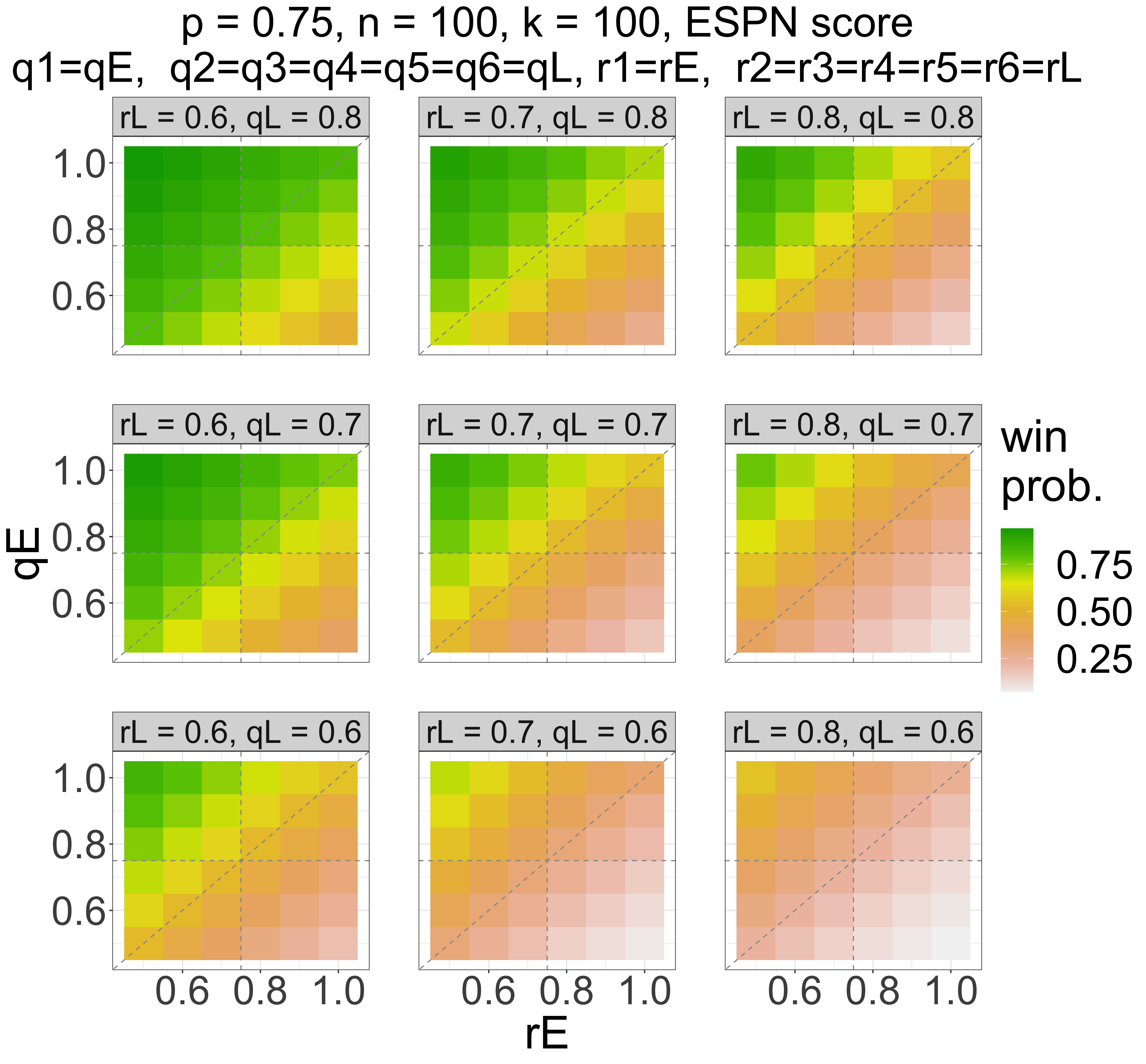}
    \label{fig:plot_smm_rd_wp_2}
} 
\caption{
    The probability (color) that the maximum ESPN score of $n$ bitstrings, with Bernoulli($q_E$) bits in early rounds and Bernoulli($q_L$) bits in later rounds, relative to a reference Bernoulli($p$) bitstring exceeds that of $k$ opposing bitstrings, with Bernoulli($r_E$) bits in early rounds and Bernoulli($r_L$) bits in later rounds, as a function of $q_E$ ($y$-axis), $r_E$ ($x$-axis), $q_L$ (rows), and $r_L$ (columns) for $p=0.75$, $k=100$. and $n=100$ in the ``guessing a randomly drawn bitstring'' contest with $R=6$ rounds.
    Figure (a) uses the partition where the first three rounds are the early rounds (e.g., $q_E = q_1 = q_2 = q_3$ and $r_E = r_1 = r_2 = r_3$) and Figure (b) uses the partition where just the first round is an early round (e.g., $q_E = q_1$ and $r_E = r_1$). 
    We should still increase the entropy of our bracket predictions as our opponents increase entropy.
}
\label{fig:plot_smm_wpScore_rd}
\end{figure}

In Figure~\ref{fig:plot_smm_wpScore_rd} we visualize win probability as a function of $q_E$, $q_L$, $r_E$, and $r_L$, for $n = k = 100$, $p = 0.75$, and ESPN score.
Figure~\ref{fig:plot_smm_rd_wp_1} uses the partition where the first three rounds are the early rounds (e.g., $q_E = q_1 = q_2 = q_3$ and $r_E = r_1 = r_2 = r_3$).
In this scenario, early round strategy $q_E$ and $r_E$ is much more impactful than late round strategy $q_L$ and $r_L$.
We visualize this as each sub-plot looking the same. 
The green triangle within each subplot illustrates that we should increase early round entropy (decrease $q_E$) as our opponents' early round entropy increases (i.e., as $r_E$ decreases).
Figure~\ref{fig:plot_smm_rd_wp_2} uses the partition where just the first round is an early round (e.g., $q_E = q_1$ and $r_E = r_1$). 
In this scenario, both early round strategy $q_E$ and $r_E$ and late round strategy $q_L$ and $r_L$ are impactful.
The green triangle appears again in each suplot, illustrating that we should increase early round entropy as our opponents' early round entropy increases.
But the green triangle grows as $r_L$ decreases, indicating that we should increase late round entropy (decrease $q_E$) as our opponent's entropy increases.

\section{An information theoretic view of the multi-brackets problem}\label{sec:ep_sec}



The multi-brackets problem is intimately connected to Information Theory.
Viewing the multi-brackets problem under an information theoretic lens provides a deeper understanding of the problem and elucidates why certain entropy-based strategies work.
In particular, the Asymptotic Equipartition Property from Information Theory helps us understand why it makes sense to increase entropy as the number of brackets increases and as our opponents' entropy increases.
In this section we give an intuitive explanation of the Equipartition Property and discuss implications, relegating the formal mathematical details to Appendix~\ref{app:EP_thm}.

To begin, we partition the set of all brackets $\X$ into three subsets, 
\begin{equation}
    \begin{cases}
    \text{low entropy ``chalky'' brackets} \ \ \qquad \C \subset \X, \\
    \text{``typical'' brackets} \qquad\qquad\qquad\qquad \T \subset \X, \\
    \text{high entropy ``rare'' brackets} \ \qquad \quad \R \subset \X.
    \end{cases}
\end{equation}
We visualize this partition of $\X$ under three lenses in Figure~\ref{fig:3_sets_3_lenses}. 

\begin{figure}[htb!]
    \centering
    \includegraphics[width=10cm]{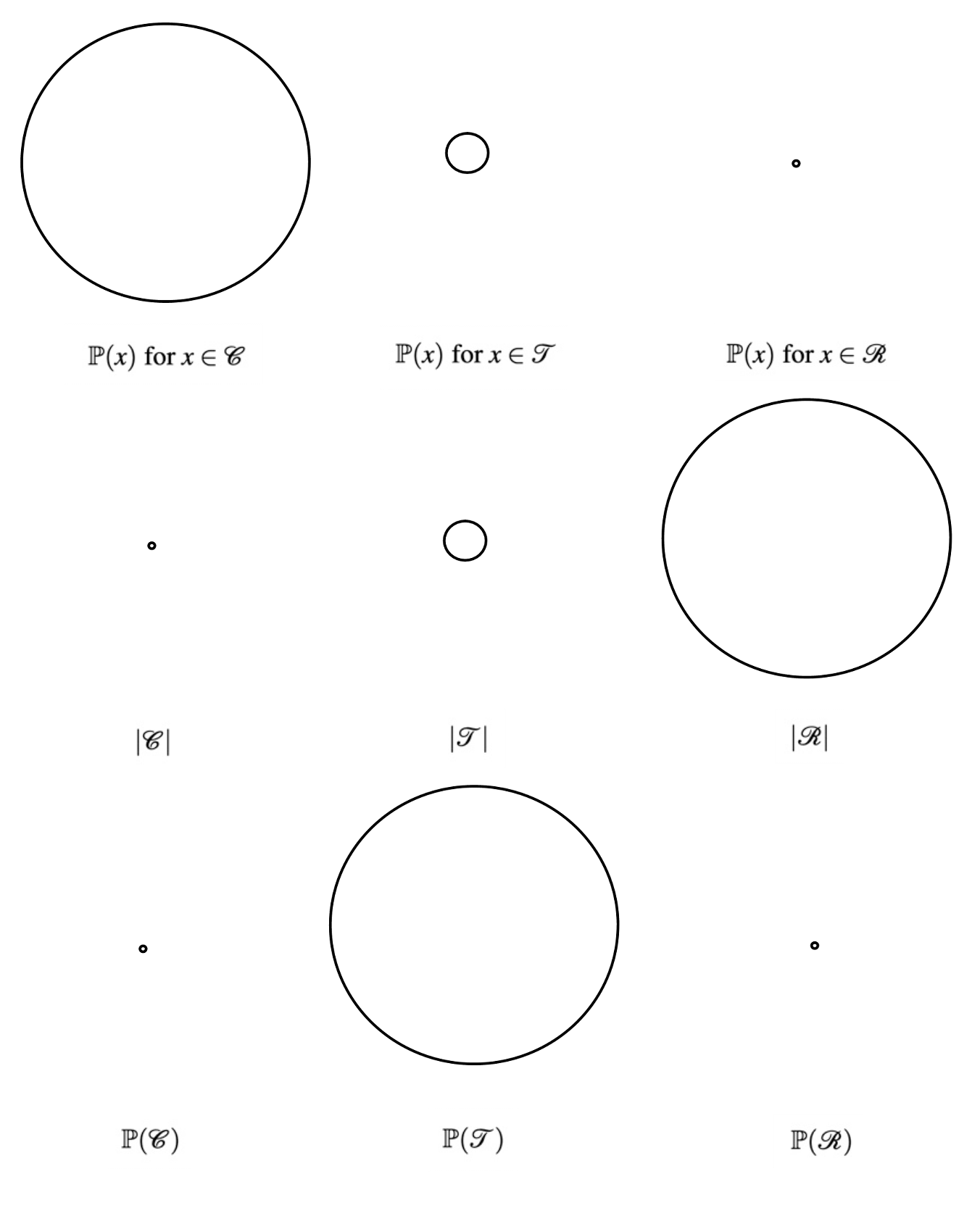}
    \caption{
        Note that these figures are not drawn to scale.
        First line: the probability mass of an individual low entropy (chalky) bracket is much larger than the probability mass of an individual typical bracket, which is much larger than the probability mass of an individual high entropy (rare) bracket. Second line:  there are exponentially more rare brackets than typical brackets, and there are exponentially more typical brackets than chalky brackets. Third line: the typical brackets occupy most of the probability mass on aggregate.
    }
    \label{fig:3_sets_3_lenses}
\end{figure}

First, the probability mass of an individual low entropy or ``chalky'' bracket is much larger than the probability mass of an individual typical bracket, which is much larger than the probability mass of an individual high entropy or ``rare'' bracket. 
In symbols, if $x_1 \in \C, x_2 \in \T, \text{ and } x_3 \in \R$, then $\P(x_1) >> \P(x_2) >> \P(x_3)$.
``Rare'' is a good name for high entropy brackets because they are highly unlikely.
``Chalk'', a term from sports betting, is a good name for low entropy brackets because it refers to betting on the heavy favorite (i.e., the outcome with highest individual likelihood).
Most of the individual forecasts within a low entropy bracket must consist of the most probable outcomes.
For example, in the ``guessing a bitstring'' contest, assuming the reference bitstring consists of independent Bernoulli($p$) bits where $p > 0.5$, low entropy brackets are bitstrings consisting mostly of ones.
In real world examples of multi-bracket pools, people are drawn to these low entropy chalky brackets because they have high individual likelihoods.

Second, there are exponentially more rare brackets than typical brackets, and there are exponentially more typical brackets than chalky brackets. In symbols, $|\R| >> |\T| >> |\C|$.
In the ``guessing a bitstring'' contest with $p > 0.5$, the overwhelming majority of possible brackets are high entropy brackets having too many zeros, and very few possible brackets are low entropy brackets consisting almost entirely of ones.
Typical brackets tow the line, having the ``right'' amount of ones.
March Madness is analagous: the overwhelming majority of possible brackets are rare brackets with too may upsets (e.g., a seed above $8$ winning the tournament) and relatively few possible brackets are chalky brackets with few upsets (there are only so many distinct brackets with favorites winning nearly all the games).
Typical brackets tow the line, having the ``right'' number of upsets. 

Lastly, the typical set of brackets contains most of the probability mass.
In symbols, $\P(\T) >> \P(\C)$ and $\P(\T) >> \P(\R)$.
This is a consequence of the previous two inequalities. 
Although $|\R|$ is massive, $\P(x)$ for $x \in \R$ is so small that $\P(\R)$ is small. 
Also, although $\P(x)$ for $x \in \C$ is relatively large, $|\C|$ is so small that $\P(\C)$ is small. 
Hence, the remainder of the probability mass, $\P(\T)$, is large.
``Typical'' is thus a good name for brackets whose entropy isn't too high or too low because a randomly drawn bracket typically has this ``right'' amount of entropy.
For example, the observed March Madness tournament is almost always a typical bracket featuring a ``typical'' number of upsets.

Drilled down to its essence, the Equipartition Property tells us that, as the number of forecasts $m$ within each bracket grows, the probability mass of the set of brackets becomes increasingly more concentrated in an exponentially small set, the ``typical set.''
See Appendix~\ref{app:EP_thm} for a more formal treatment of the Equipartition Property.

This information theoretic view of the multi-brackets problem sets up a tradeoff between chalky and typical brackets.
Typical brackets have the ``right'' entropy but consist of less likely individual outcomes, whereas chalky low entropy brackets have the ``wrong'' entropy but consist of more likely individual outcomes.
The former excels when $n$ is large, the latter excels when $n$ is small, and for moderate $n$ we interpolate between these two regimes; so, we should increase the entropy of our set of predicted brackets as the number of brackets $n$ increases.
We justify this below using the Equipartition Property. 

As the typical set contains most of the probability mass, the reference bracket is highly likely a typical bracket.
So when $n$ is large we should generate typical brackets as guesses since it is likely that at least one of these guesses is close to the reference bracket.
When $n$ is small, generating typical brackets as guesses doesn't produce as high an expected maximum score as chalky brackets.
To understand, recall that a bracket consists of $m$ individual forecasts.
A single randomly drawn typical bracket has the same entropy as the reference bracket but isn't likely to correctly predict each individual forecast.
For instance, in our ``guessing a bitstring'' example, a single randomly drawn bitstring has on average a similar number of ones as the reference bitstring, but not the \textit{right} ones in the right locations.
A chalky bracket, on the other hand, predicts highly likely outcomes in most of the individual forecasts.
The chalkiest bracket, which predicts the most likely outcome in each individual forecast, matches the reference bracket for each forecast in which the reference bracket realizes its most likely outcome.
This on average yields more matches than that of a typical bracket because more forecasts realize their most likely outcome than any other single outcome.
For instance, in our ``guessing a bitstring'' example, a chalky bracket consists mostly of ones (assuming $p > 0.5$) and so correctly guesses the locations of ones in the reference bitstring.
This is better on average than guessing a typical bracket, which has on average has the right number of ones but in the wrong locations.


\section{Real world examples}\label{sec:realWorldExamples}

Now, we discuss real world examples of multi-bracket pools: pick six betting in horse racing and March Madness bracket challenges.
Both contests involve predicting a tuple of outcomes. 
An individual pick six bet (ticket) involves predicting the winner of each of six horse races and an individual March Madness bet (bracket) involves predicting the winner of each game in the NCAA Division I Men's Basketball ``March Madness'' tournament.
In both contests it is allowed, but not necessarily commonplace (outside of horse racing betting syndicates), to submit many tickets or brackets.
We demonstrate that the entropy-based strategies introduced in the previous sections are particularly well-suited for these problems.
In particular, optimizing over strategies of varying levels of entropy is tractable and yields well-performing solutions.

\subsection{Pick six horse race betting}\label{sec:ex_pickSix}

Horse race betting is replete with examples of multi-bracket pools.
A prime example is the \textit{pick six} bet, which involves correctly picking the winner of six horse races.
Similar pick three, pick four, and pick five bets, which involve correctly picking the winner of three, four, or five horse races, respectively, also exist.
Due to the immense difficulty of picking six consecutive horse race winners coupled with a large number of bettors in these pools, payoffs for successful pick six bets can be massive (e.g., in the millions of dollars).
In this section we apply our entropy-based strategies to pick six betting, demonstrating the massive profit potential of these bets.

To begin, let $s \in \{3,4,5,6\}$ denote the number of races comprising the pick-$s$ bet (for the pick three, four, five, and six contests, respectively).
Suppose for simplicity that one pick-$s$ ticket, consisting of $s$ predicted horse race winners, costs $\$1$ each (typically a pick-$s$ bet costs $\$1$ or $\$2$).
Indexing each race by $j=1,...,s$, suppose there are $m_j$ horses in race $j$, and let $\bm = (m_1,...,m_s)$.
There is a fixed carryover $C$, an amount of money leftover from previous betting pools in which no one won, that is added to the total prize pool for the pick-$s$ contest.
As done throughout this paper, assume the true win probability $\bP_{ij}$ that horse $i$ wins race $j$ is known for each $i$ and $j$.
As our operating example in this section, we set $\bP$ to be the win probabilities implied by the Vegas odds from the pick six contest from Belmont Park on May 21, 2023,\footnote{
    \url{https://entries.horseracingnation.com/entries-results/belmont-park/2023-05-21}
} which we visualize in Figure~\ref{fig:P_picksix}.

\begin{figure}[htb!]
    \centering
    \includegraphics[width=10cm]{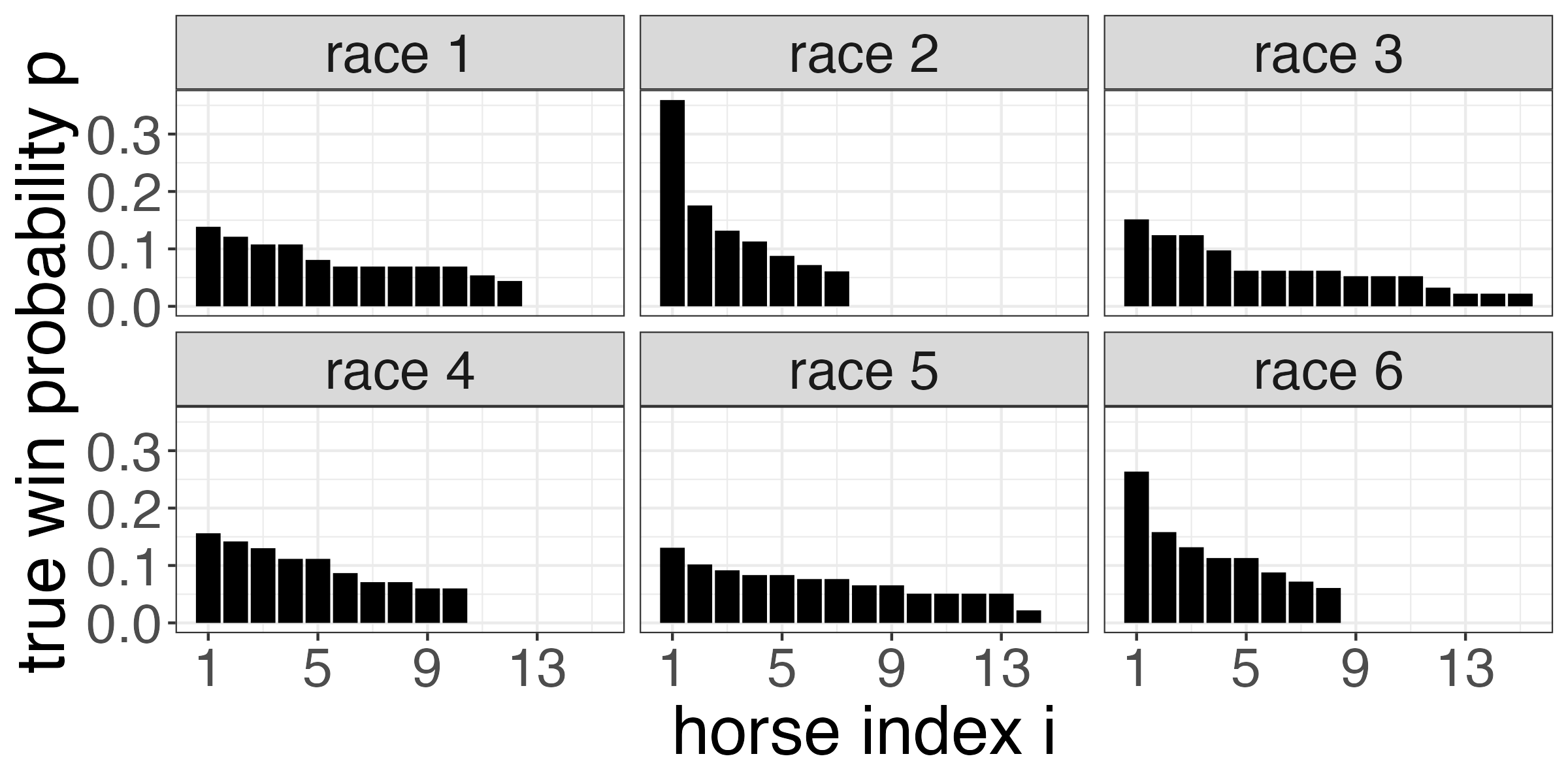}
    \caption{
        The ``true'' win probability $\bP_{ij}$ ($y$-axis) that horse $i$ ($x$-axis) wins race $j$ (facet) for each $i,j$.  These probabilities are implied by the Vegas odds for the pick six contest at Belmont Park on May 21, 2023.
    }
    \label{fig:P_picksix}
\end{figure}

Suppose the public purchases $k$ entries according to some strategy.
In particular, we assume the public submits $k$ independent tickets according to $\bR$, where $\bR_{ij}$ is the probability an opponent selects horse $i$ to win race $j$.
We purchase $n$ entries according to strategy $\bQ$. 
Specifically, we submit $n$ independent tickets according to $\bQ$, where $\bQ_{ij}$ is the probability we select horse $i$ to win race $j$.
The total prize money is thus 
\begin{equation}
    T = C + (n+k)(1-\alpha),
\label{eqn:pick_six_total_money}
\end{equation}
where $\alpha$ is the track take (vigorish).
Let $W$ be our number of winning tickets and let $\Wopp$ be our opponents' number of winning tickets.
Under our model, both $W$ and $\Wopp$ are random variables.
Formally, denote the ``true'' observed $s$ winning horses by $\tau = (\tau_1,...,\tau_s)$, our $n$ tickets by $(x^{(1)},...,x^{(n)})$ where each $x^{(\ell)} = (x^{(\ell)}_1,...,x^{(\ell)}_s)$, and the publics' $k$ tickets by $(y^{(1)},...,y^{(k)})$ where each  $y^{(\ell)} = (y^{(\ell)}_1,...,y^{(\ell)}_s)$.
Then
\begin{equation}
    W = \sum_{\ell=1}^{n} \one\{x^{(\ell)} = \tau\} = \sum_{\ell=1}^{n} \one\{x^{(\ell)}_1 = \tau_1,...,x^{(\ell)}_s = \tau_s\}
\label{eqn:pick_six_W}
\end{equation}
and
\begin{equation}
    \Wopp = \sum_{\ell=1}^{k} \one\{y^{(\ell)} = \tau\} = \sum_{\ell=1}^{k} \one\{y^{(\ell)}_1 = \tau_1,...,y^{(\ell)}_s = \tau_s\}.
\label{eqn:pick_six_W_opp}
\end{equation}
Then the amount we profit is also a random variable,
\begin{equation}
    \text{Profit} = \bigg( \frac{W}{W + \Wopp} \bigg)T - n,
\label{eqn:pick_six_profit}
\end{equation}
where we treat $\frac{0}{0}$ to be 0 (i.e., if both $W=0$ and $\Wopp=0$, the fraction $W/(W+\Wopp)$ is 0).
Here, the randomness is over $\tau \sim \bP$, $x \sim \bQ$, and $y \sim \bR$.

Our task is to solve for the optimal investment strategy $\bQ$ given all the other variables $n$, $k$, $\bP$, $\bR$, $C$, and $\alpha$.
Formally, we wish to maximize expected profit,
\begin{equation}
    \E[\text{Profit}] = -n + T\cdot\E\bigg( \frac{W}{W + \Wopp} \bigg).
\label{eqn:pick_six_Eprofit}
\end{equation}
In Appendix~\ref{app:pick_six_details} we compute a tractable lower bound for the expected profit.

\begin{figure}[htb!]
    \centering
    \includegraphics[width=10cm]{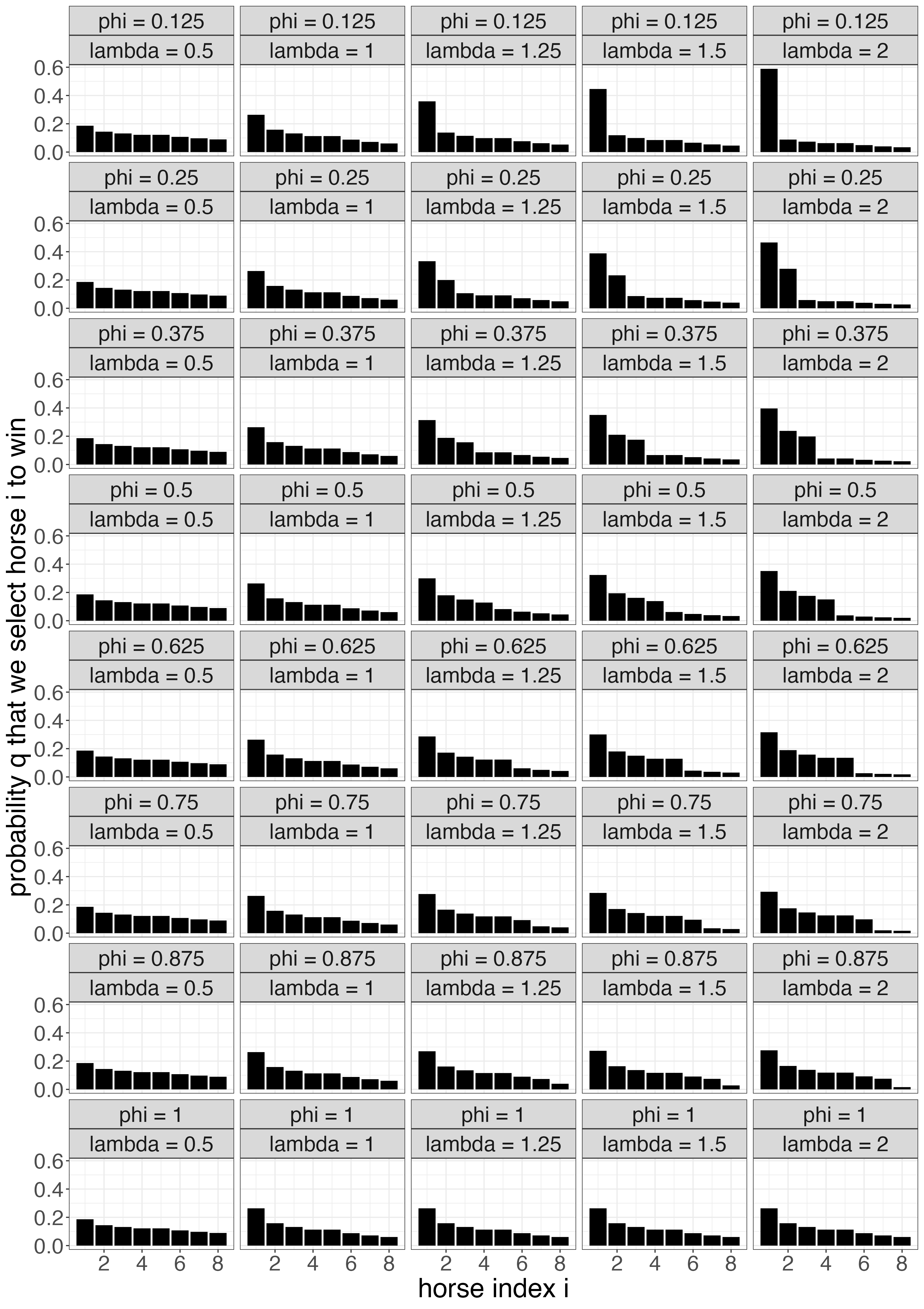}
    \caption{
        The probability $\bQ_{ij} = \bQ_{ij}(\lambda, \phi)$ ($y$-axis) that we select horse $i$ ($x$-axis) to win race $j = 6$ for various values of $\lambda$ (column) and $\phi$ (row).
        For fixed $\phi$, entropy increases as $\lambda$ decreases.
        For fixed $\lambda$, the probabilities of successively fewer horses are upweighted as $\phi$ decreases.
    }
    \label{fig:tilted_P_picksix}
\end{figure}

We are unable to analytically optimize the expected profit to find an optimal strategy $\bQ^\ast$ given the other variables, and we are unable to search over the entire high dimensional $\bQ$-space for an optimal strategy.
Instead, we apply the entropy-based strategies described in the previous sections.
The idea is to search over a subspace of $\bQ$ that explores strategies of varying entropies, finding the optimal entropy given the other variables.
To generate $n$ pick six tickets at varying levels of entropy, we let $\bQ = \bQ(\lambda, \phi)$ vary according to parameters $\lambda$ and $\phi$ that control the entropy.
Assuming without loss of generality that in each race $j$ the true win probabilities are sorted in decreasing order, $\bP_{1j} \geq \bP_{2j} \geq ... \geq \bP_{m_jj}$, we define $\bQ(\lambda, \phi)$ for $\lambda > 0$ and $\phi \in [0,1]$ by 
\begin{align}
    \widetilde\bQ_{ij}(\lambda, \phi) &= 
    \begin{cases}
        \big(\bP_{ij}/\bP_{round(\phi \cdot m_j), j}\big)^\lambda & \text{ if } \lambda < 1, \\
        \bP_{ij} \big( \lambda\cdot\one\{i \leq round(\phi \cdot m_j)\} + \frac{1}{\lambda} \cdot \one\{i \leq round(\phi \cdot m_j)\} \big) & \text{ if } \lambda \geq 1,
    \end{cases} \\
    \bQ_{ij}(\lambda, \phi) &= \widetilde\bQ_{ij}(\lambda, \phi) / \sum_{i=1}^{m_j} \widetilde\bQ_{ij}(\lambda, \phi),
\label{eqn:Q_lambda_phi_pick_six}
\end{align}
recalling that there are $m_j$ horses in race $j$.
We visualize these probabilities for race $j=6$ in Figure~\ref{fig:tilted_P_picksix}.
For fixed $\phi$, smaller values of $\lambda$ push the distribution $\bQ_{\ast j}$ closer towards the uniform distribution, increasing its entropy.
Conversely, increasing $\lambda$ lowers its entropy.
In lowering its entropy, we shift the probability from some horses onto other horses in a way that makes the distribution less uniform.
The parameter $\phi$ controls the number of horses to which we transfer probability as $\lambda$ increases.
For instance, there are $m_j = 8$ horses in race $j = 6$, so when $\phi = 3/8$ we transfer successively more probability to the top $3 = round(\phi \cdot m_j)$ horses as $\lambda$ increases.


Further, we assume we play against opponents who generate brackets according to the strategy $\bR_{ij}(\lambda_{opp}) = \bP(\lambda = \lambda_{opp}, \phi = 1/8)$.
In other words, low entropy opponents bet mostly on the one or two favorite horses (depending on $m_j$), high entropy opponents are close to the uniform distribution, and moderate entropy opponents lie somewhere in the middle.
The exact specification of the opponents' distribution isn't important, as we use it to illustrate a general point.
In future work, one can try to model the distribution of the publics' ticket submissions to get more precise results.

In Figure~\ref{fig:eProfit_pickSix} we visualize expected profit for a pick six horse racing betting pool in which we submit $n$ tickets according to strategy $\bQ(\lambda,\phi)$ against a field of $k = 25,000$ opponents who use strategy $\bR(\lambda_{opp})$, assuming a track take of $\alpha = 0.05$ and carryover $C=500,000$, as a function of $\lambda_{opp}$ and $n$.
Given these variables, we use the strategy $(\lambda, \phi)$ that maximizes expected profit over a grid of values.
We see that the entropy of the optimal strategy increases as $n$ increases (i.e., $\lambda$ decreases and $\phi$ increases as $n$ increases).
Further, we see that submitting many brackets at a smart entropy level is hugely profitable.
This holds true particularly when the carryover is large enough, which occurs fairly regularly.

\begin{figure}[htb!]
    \centering
    \includegraphics[width=15cm]{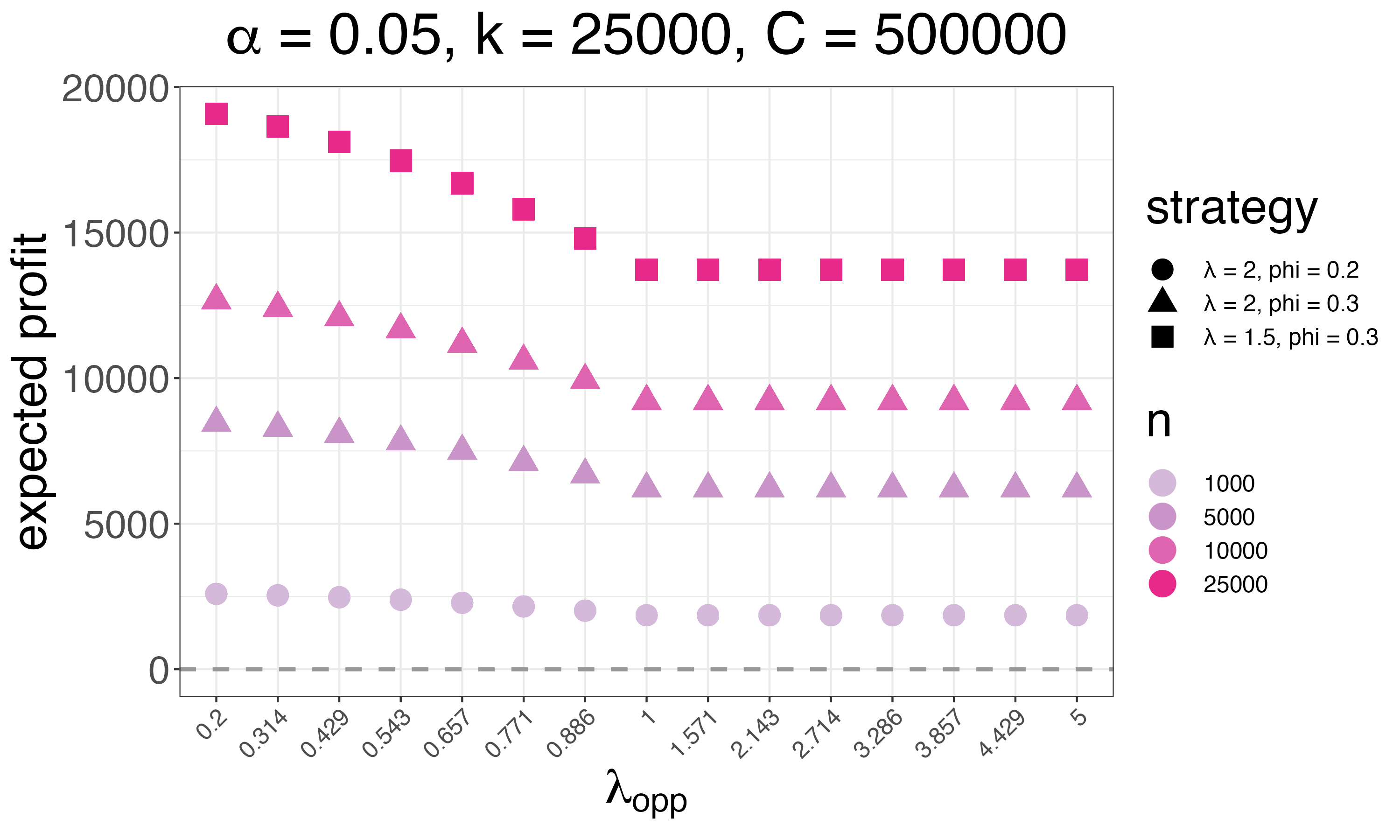}
    \caption{
        A lower bound of our expected profit ($y$-axis) for a pick six horse racing betting pool in which we submit $n$ tickets according to strategy $\bQ(\lambda,\phi)$ against a field of $k = 25,000$ opponents who use strategy $\bR(\lambda_{opp})$, assuming a track take of $\alpha = 0.05$ and carryover $C=500,000$, as a function of $\lambda_{opp}$ ($x$-axis) and $n$ (color).
        Given these variables, we use the strategy $(\lambda, \phi)$ that maximizes expected profit over a grid of values.
    }
    \label{fig:eProfit_pickSix}
\end{figure}

\subsection{March Madness bracket challenge}\label{sec:ex_MM}

March Madness bracket challenges are prime examples of multi-bracket pools.
In a bracket challenge, contestants submit an entire \textit{bracket}, or a complete specification of the game winners of each of the games in the NCAA Division I Men's Basketball ``March Madness'' tournament.
The winning bracket is closest to the observed NCAA tournament according to some metric.
Popular March Madness bracket challenges from ESPN, BetMGM, and DraftKings, for instance, offer large cash prizes -- BetMGM offered $\$10$ million to a perfect bracket or $\$100,000$ dollars to the closest bracket, DraftKings sent $\$60,000$ in cash prizes spread across the best $5,096$ brackets last year, and ESPN offered $\$100,000$ to the winner of a lottery among the entrants who scored the most points in each round of the tournament.\footnote{
    \url{https://www.thelines.com/best-march-madness-bracket-contests/}
}
To illustrate the difficulty of perfectly guessing the observed NCAA tournament, Warren Buffett famously offered $\$1$ billion to anyone who filled out a flawless bracket.\footnote{
    \url{https://bleacherreport.com/articles/1931210-warren-buffet-will-pay-1-billion-to-} \\ \url{fan-with-perfect-march-madness-bracket}
}
In this section we apply our entropy-based strategies to March Madness bracket challenges, demonstrating the impressive efficacy of this strategy.

To begin, denote the set of all brackets by $\X$, which consists of $N = 2^{63} = 2^{2^{6}-1}$ brackets since there are $63$ games through $6$ rounds in the NCAA tournament (excluding the four game play-in tournament). We define an atomic probability measure $\P$ on $\X$, where $\P(x)$ is the probability that bracket $x \in \X$ is the ``true'' observed NCAA tournament, as follows. 
Given that match $m \in \{1,...,63\}$ involves teams $i$ and $j$, we model the outcome of this match by $i \cdot b_m + j \cdot (1-b_m)$ where $b_m \overset{ind}{\sim} \text{Bernoulli}(\mathbf{P}_{ij})$.
In other words, with probability $\mathbf{P}_{ij}$, team $i$ wins the match, else team $j$ wins the match. 
Prior to the first round (games 1 through 32), the first 32 matchups are set. Given these matchups, the 32 winning teams in round one are determined by Bernoulli coin flips according to $\mathbf{P}$. These 32 winning teams from round one then uniquely determine the 16 matchups for the second round of the tournament. Given these matchups, the 16 winning teams in round two are also determined by Bernoulli coin flips according to $\mathbf{P}$. These winners then uniquely determine the matchups for round three. This process continues until the end of round six, when one winning team remains.

In this work, we assume we know the ``true'' win probabilities $\mathbf{P}$.
As our operating example in this section, we set $\bP$ to be the win probabilities implied by FiveThirtyEight's Elo ratings from the 2021 March Madness tournament.\footnote{
    \url{https://projects.fivethirtyeight.com/2022-march-madness-predictions/}
}
We scrape FiveThirtyEight's pre-round-one 2021 Elo ratings $\{\beta_i\}_{i=1}^{64}$ and index the teams by $i \in \{1,...,64\}$ in decreasing order of Elo rating (e.g., the best team Gonzaga is 1 and the worst team Texas Southern is 64).
Then we define $\mathbf{P}$ by $\mathbf{P}_{ij} = 1/(1+10^{-(\beta_i - \beta_j)*30.464/400})$.
In Figure~\ref{fig:ELO538A} we visualize $\{\beta_i\}_{i=1}^{64}$.
The Elo ratings range from 71.1 (Texas Southern) to 96.5 (Gonzaga), who is rated particularly highly.
In Figure~\ref{fig:ELO538B} we visualize $\mathbf{P}$ via the functions $j \mapsto \mathbf{P}_{ij}$ for each team $i$. 
For instance, Gonzaga's win probability function is the uppermost orange line, which is considerably higher than the other teams' lines. 

\begin{figure}[htb!]
\centering
\subfloat[Histogram of Elo Ratings]{\includegraphics[width = .4\textwidth]{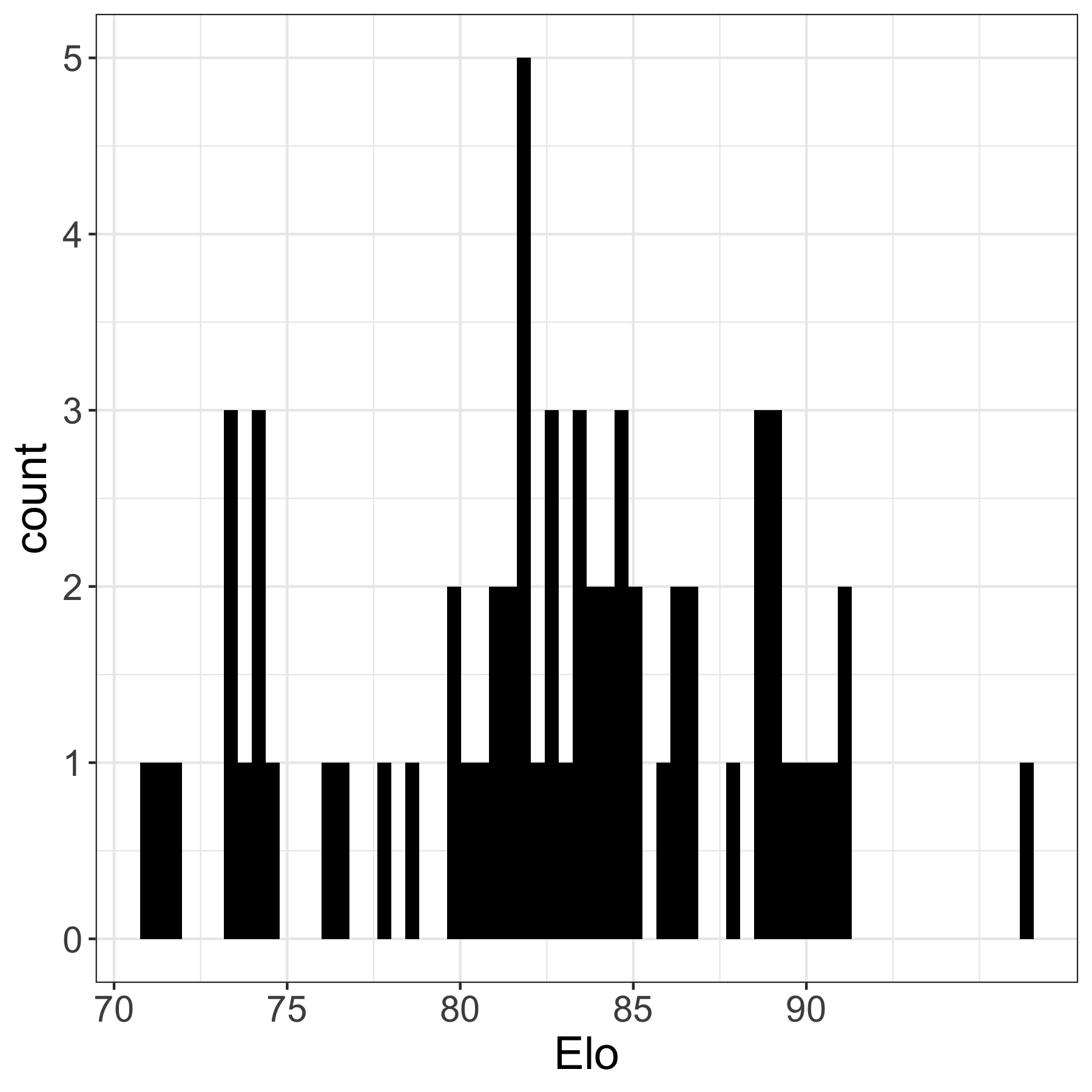}\label{fig:ELO538A}} \qquad
\subfloat[$j \mapsto \mathbf{P}_{ij}$ for each $i$]{\includegraphics[width = .4\textwidth]{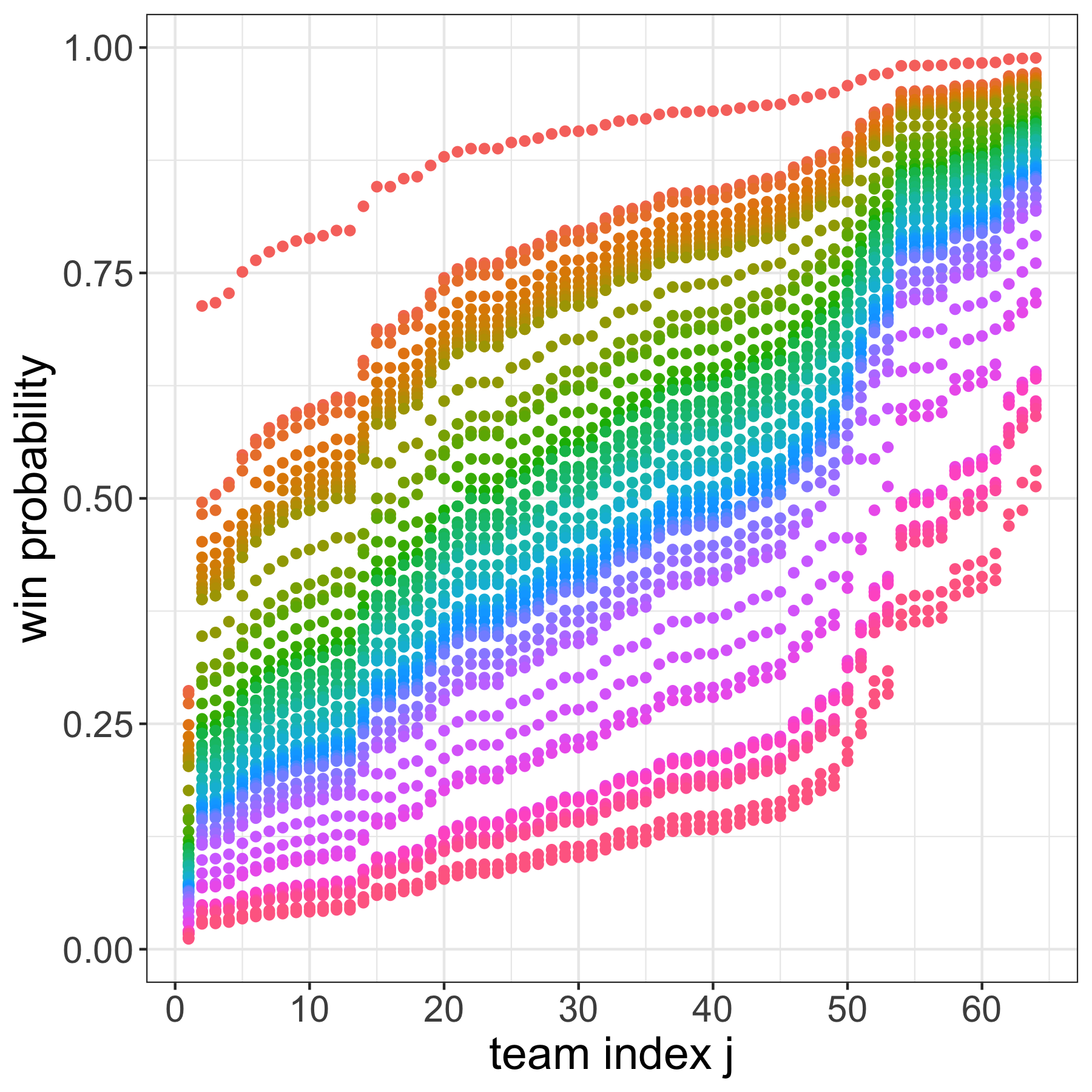}\label{fig:ELO538B}} 
\caption{
    Figure (a): histogram of FiveThirtyEight's pre-round-one Elo ratings for the 2021 March Madness tournament.
    Figure (b): the function $j \mapsto \mathbf{P}_{ij}$ for each team $i$ (color) implied by these Elo ratings.
}
\label{fig:ELO538}
\end{figure}

Suppose a field of opponents submits $k$ brackets $(y^{(1)},...,y^{(k)}) \subset \X$ to the bracket challenge according to some strategy $\bR$.
In particular, we assume the public submits $k$ independent brackets according to $\bR$, where $\bR_{ij}$ is the probability an opponent selects team $i$ to beat team $j$ in the event that they play. 
We submit $n$ brackets $(x^{(1)},...,x^{(n)}) \subset \X$ to the bracket challenge according to strategy $\bQ$. 
Specifically, we submit $n$ independent brackets according to $\bQ$, where $\bQ_{ij}$ is the probability we select team $i$ to beat team $j$ in the event that they play. 
The goal is to get as ``close'' to the ``true'' reference bracket $\tau \in \X$, or the observed NCAA tournament, as possible according to a bracket scoring function.
The most common such scoring function in these bracket challenges is what we call \textit{ESPN score}, which credits $10 \cdot 2^{\rd-1}$ points to correctly predicting the winner of a match in round $\rd \in \{1,...,6\}$.
Since there are $2^{6 - \rd}$ matches in each round $\rd$, ESPN score ensures that the maximum accruable points in each round is the same ($320$). 
Formally, our task is to submit $n$ brackets so as to maximize the probability we don't lose the bracket challenge,
\begin{align}
    \P\bigg[ \max_{j=1,...,n} f(x^{(j)},\tau) \geq \max_{\ell=1,...,k} f(y^{(\ell)},\tau) \bigg].
    \label{eqn:MMwp}
\end{align}
Alternatively, in the absence of information about our opponents, our task is to submit $n$ brackets so as to maximize expected maximum score,
\begin{align}
    \E\bigg[\max_{j=1,...,n} f(x^{(j)},\tau)\bigg].
    \label{eqn:MMeMaxScore}
\end{align}

Under this model, it is intractable to explicitly \textit{evaluate} these formulas for expected maximum score or win probability for general $\bP$, $\bQ$, and $\bR$, even when we independently draw brackets from these distributions.
This is because the scores $f(x^{(1)},\tau)$ and $f(x^{(2)},\tau)$ of two submitted brackets $x^{(1)}$ and $x^{(2)}$ relative to $\tau$ are both dependent on $\tau$, and integrating over $\tau$ yields a sum over all $2^{m} = 2^{63}$ possible true brackets for $\tau$, which is intractable.
Hence we use Monte Carlo simulation to approximate expected maximum score and win probability.
We approximate expected maximum score via
\begin{align}
    \E\bigg[ \max_{j=1,...,n} f(x^{(j)},\tau) \bigg] 
    \approx \frac{1}{B_1} \sum_{b_1=1}^{B_1} \frac{1}{B_2} \sum_{b_2=1}^{B_2} \max_{j=1,...,n} f(x^{(j,b_2)},\tau^{(b_1)}),
    \label{eqn:fmm_monte_carlo_escore}
\end{align}
where the $\tau^{(b_1)}$ are independent samples from $\bP$ and the $x^{(j,b_2)}$ are independent samples from $\bQ$.
We use a double Monte Carlo sum, with $B_1 = 250$ draws of $\tau$ and $B_2 = 100$ draws of $(x^{(1)},...,x^{(n)})$, because it provides a smoother and stabler approximation than a single Monte Carlo sum.
Similarly, we approximate win probability via
\begin{align}
    & \P \bigg[ \max_{j=1,...,n} f(x^{(j)},\tau) \geq \max_{\ell=1,...,k} f(y^{(\ell)},\tau) \bigg] \\ 
    \approx&
    \frac{1}{B_1} \sum_{b_1=1}^{B_1} \frac{1}{B_2} \sum_{b_2=1}^{B_2} 
   \one\bigg\{ \max_{j=1,...,n} f(x^{(j,b_2)},\tau^{(b_1)}) \geq \max_{\ell=1,...,k} f(y^{(\ell,b_2)},\tau^{(b_1)}) \bigg\},
    \label{eqn:fmm_monte_carlo_wp}
\end{align}
where the $\tau^{(b_1)}$ are independent samples from $\bP$, the $x^{(j,b_2)}$ are independent samples from $\bQ$, and the $y^{(\ell,b_2)}$ are independent samples from $\bR$.
We again use a double Monte Carlo sum, with $B_1 = 250$ draws of $\tau$ and $B_2 = 100$ draws of $(x^{(1)},...,x^{(n)})$ and $(y^{(1)},...,y^{(k)})$, because it provides a smooth and stable approximation.

We are unable to analytically optimize these objective functions to find an optimal strategy $\bQ^\ast$ given the other variables, and we are unable to search over the entire high dimensional $\bQ$-space for an optimal strategy.
These problems are even more difficult than simply evaluating these objective functions, which itself is intractable.
Thus, we apply the entropy-based strategies from the previous sections, which involve generating successively higher entropy brackets as $n$ increases.
The idea is to search over a subspace of $\bQ$ that explores strategies of varying entropies, finding the optimal entropy given the other variables.
To generate $n$ brackets at varying levels of entropy, we let $\bQ = \bQ(\lambda)$ vary according to the parameter $\lambda$ that controls the entropy.
In a game in which team $i$ is favored against team $j$ (so $i < j$, since we indexed the teams in decreasing order of team strength, and $\bP_{ij} \in [0.5,1]$), the lowest entropy (chalkiest) strategy features $\bQ_{ij} = 1$, the ``true'' entropy strategy features $\bQ_{ij} = \bP_{ij}$, and the highest entropy strategy features $\bQ_{ij} = 1/2$. 
We construct a family for $\bQ$ that interpolates between these three poles,
\begin{equation}
    \bQ_{ij}(\lambda) := 
    \begin{cases}
        (1-2\lambda) \frac{1}{2} + (2\lambda) \bP_{ij}  & \text{ if } \lambda \in [0,\frac{1}{2}] \text{ and } i<j, \\
        (1-2(\lambda-\frac{1}{2})) \bP_{ij} + 2(\lambda-\frac{1}{2}) 1 & \text{ if }  \lambda \in [\frac{1}{2},1] \text{ and } i<j,
    \end{cases}
    \label{eqn:lambda_chalky_2}
\end{equation}
where $\lambda \in [0,1]$.
The entropy of $\bQ(\lambda)$ increases as $\lambda$ decreases.

Further, we assume we play against \textit{colloquially chalky} opponents, who usually bet on the higher seeded team.
Each team in the March Madness tournament is assigned a numerical ranking from 1 to 16, their \textit{seed}, prior to the start of the tournament by the NCAA Division I Men's Basketball committee.
The seeds determine the matchups in round one and are a measure of team strength (i.e., lower seeded teams are considered better by the committee).
We suppose colloquially chalky opponents generate brackets according to a distribution $\Pchalky$ based on the seeds $s_i$ and $s_j$ of teams $i$ and $j$,
\begin{equation}
\Pchalky_{ij} = 
\begin{cases}
    0.9 \text{ if } s_i - s_j < -1, \\
    0.5 \text{ if } |s_i - s_j| \leq 1, \\
    0.1 \text{ if } s_i - s_j > 1,
\end{cases}
\label{eqn:colloquially_chalky}
\end{equation}
so they usually bet on the higher seeded team.
The exact specification of the colloquially-chalky distribution isn't important, as we use $\Pchalky$ to illustrate a general point.
In future work, one can try to model the distribution of the publics' bracket submissions to get more precise results.


In Figure~\ref{fig:plotMMeMaxScore} we visualize the expected max score of $n$ brackets generated according to $\bQ(\lambda)$ as a function of $n$ and $\lambda$.
In Figure~\ref{fig:plotMMwp} we visualize the probability that the max score of $n$ brackets generated according to $\bQ(\lambda)$ exceeds that of $k = 10,000$ colloqually chalky brackets generated according to $\Pchalky$ as a function of $n$ and $\lambda$.
In both, we again see that we should increase entropy (decrease $\lambda$) as $n$ increases.
In particular, the small circle (indicating the best strategy given $n$ and $k$) moves leftward as $n$ increases.
Further, we see that tuning the entropy of our submitted bracket set given the other variables yields an excellent win probability, even when $n$ is much smaller than $k$. 

\begin{figure}[htb!]
\centering
\subfloat[]{\includegraphics[width = .45\textwidth]{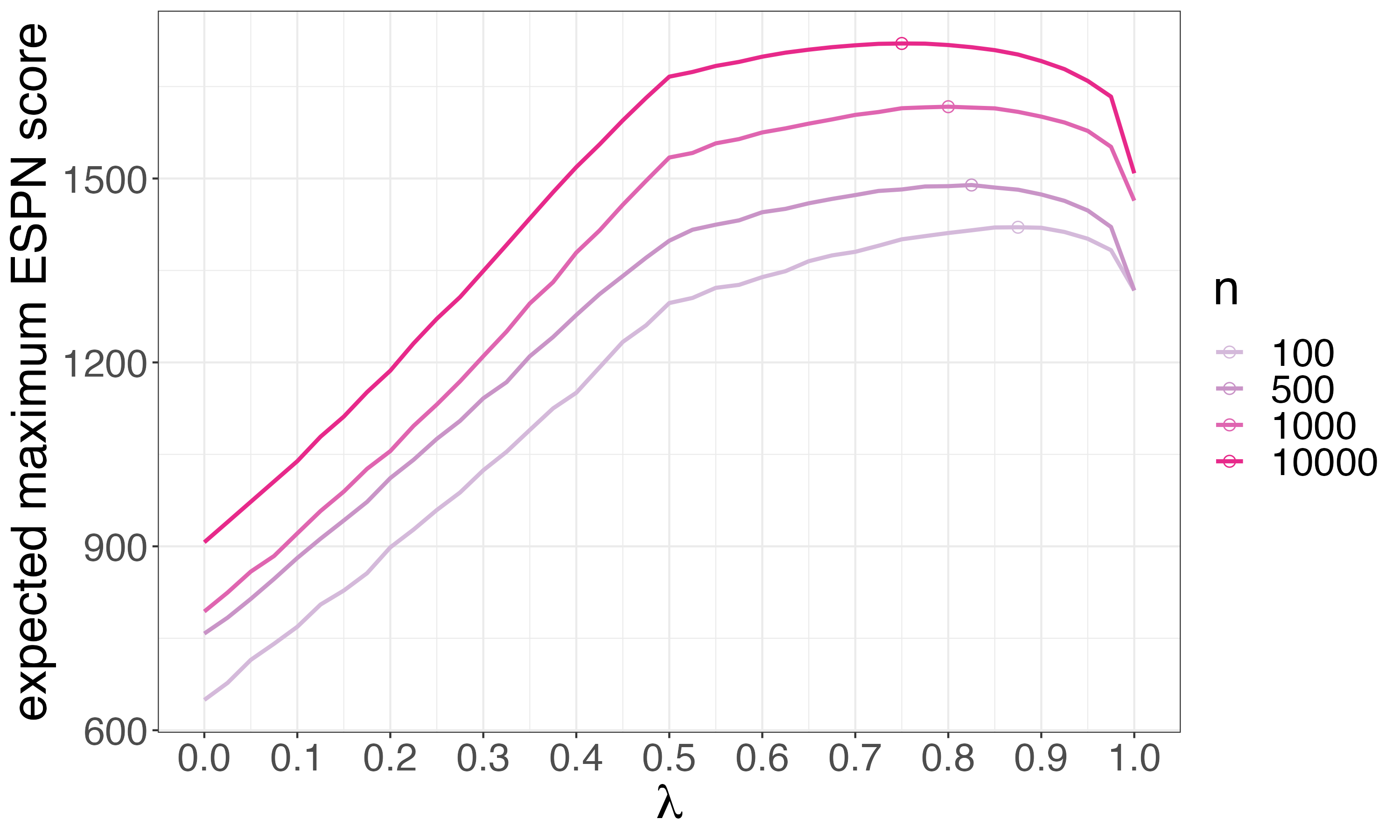}\label{fig:plotMMeMaxScore}} \qquad
\subfloat[]{\includegraphics[width = .45\textwidth]{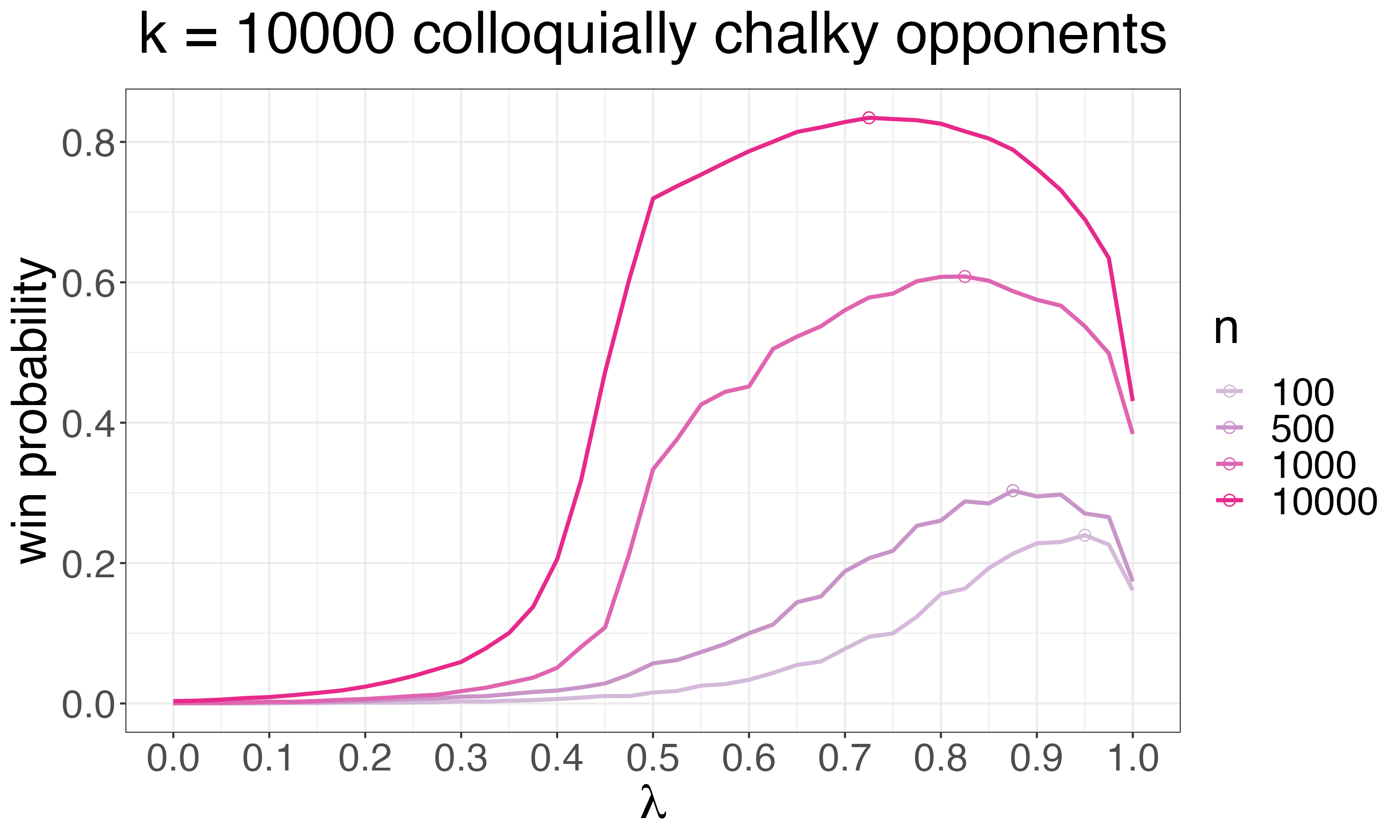}\label{fig:plotMMwp}} 
\caption{
    Figure (a): the expected max ESPN score ($y$-axis) of $n$ brackets generated according to $\bQ(\lambda)$ as a function of $n$ (color) and $\lambda$ ($x$-axis).
    Figure (b): the probability ($y$-axis) that the max ESPN score of $n$ brackets generated according to $\bQ(\lambda)$ exceeds that of $k = 10,000$ colloqually chalky brackets generated according to $\Pchalky$ as a function of $n$ (color) and $\lambda$ ($x$-axis).
    The small circle indicates the best strategy given $n$ and $k$.
    We want to increase entropy (decrease $\lambda$) as $n$ increases.
}
\label{fig:plotMMresults}
\end{figure}

\section{Discussion}\label{sec:discussion}

In this work, we pose and explore the multi-brackets problem: how should we submit $n$ predictions of a randomly drawn reference bracket (tuple)? 
The most general version of this question, which finds the optimal set of $n$ brackets across all such possible sets, is extremely difficult.
To make the problem tractable, possible, and/or able to be visualized, depending on the particular specification of the multi-bracket pool, we make simpifying assumptions.
First, we assume we (and optionally a field of opponents) submit i.i.d. brackets generated according to a bracket distribution.
The task becomes to find the optimal generating bracket distribution.
For some multi-bracket pools this is tractable and for others it is not.
For those pools, we make another simplifying assumption, searching over a smartly chosen low dimensional subspace of generating bracket distributions covering distributions of various levels of entropy.
We find this approach is sufficient to generate well-performing sets of submitted brackets.
We also learn the following high-level lessons from this strategy: we should increase the entropy of our bracket predictions as $n$ increases and as our opponents increase entropy.

We leave much room for future work on the multi-brackets problem.
First, it is still an open and difficult problem to find the optimal set of $n$ bracket predictions across \textit{all} such possible subsets, where optimal could mean maximizing expected maximum score, win probability, or expected profit.
Second, in this work we assume the ``true'' probabilities $\bP$ and our opponents' generating bracket strategy $\bR$ exists and are known.
A fruitful extension of this work would revisit the problems posed in this work under the lens that, in practice, these distributions are either estimated from data or are unknown (e.g., as in \citet{metel2017}).
Finally, we suggest exploring more problem-specific approaches to particular multi-bracket pools.
For instance, in March Madness bracket challenges we suggest exploring strategies of varying levels entropy within each round.
Perhaps the publics' entropy is too low in early rounds and too high in later rounds, suggesting we should counter by increasing our entropy in earlier rounds and decreasing our entropy in later rounds.



\bibliography{ref.bib}

\newpage
\begin{center}
{\LARGE \textbf{Appendix}}
\end{center}
\appendix




\section{Our code}\label{app:code}

Our code for this project is publicly available at \url{https://github.com/snoopryan123/entropy_ncaa}.

\section{Previous work details}\label{app:prevWork_details}

\subsection{Estimating outcome probabilities in horse racing and March Madness}\label{app:probEstimationDetails}

There has been a plethora of research on estimating horse race outcome probabilities.
A line of research beginning with \citet{harville1973} estimates the probabilities of the various possible orders of finish of a horse race assuming knowledge of just the win probabilities of each individual horse.
\citet{Henery1981}, \citet{stern1990}, \citet{Lo1992}, \citet{Lo1994}, and \citet{LO2008} extend this work, developing better and more tractable models.
There has also been extensive research on the favorite-longshot bias, the phenomenon that the public typically underbets favored horses and overbets longshot horses, skewing the win probabilities implied by the odds that each horse wins the race.
For instance, \citet{Asch1984}, \citet{Ali1998}, \citet{ROSENBLOOM2003}, \citet{Brown2003}, and \citet{Hausch2008} illustrate and explain the favorite-longshot bias, and some of these works attempt to adjust for this bias in estimating horse racing outcome probabilities.
Yet another line of research focuses on comprehensively estimating these horse finishing probabilities.
\citet{bolton1986} model outcome probabilities using a multinomial logistic regression model, forming the basis for most modern prediction methods.
\citet{chapman2008}, \citet{benter}, \citet{edelman2007}, \citet{Lessmann2007}, and \citet{Lessmann2009} extend this work.\footnote{
    Benter in particular reported that his team has made significant profits during their five year gambling operation. 
}
\citet{Dayes2010} searches for covariates that are predictive of horse race outcome probabilities even after adjusting for the odds, \citet{Silverman2012} estimates a hierarchical Bayesian model of these probabilities, and \citet{Gulum2018} uses machine learning and graph-based features to estimate these probabilities.
For a more comprehensive review of this literature, see \citet{Hausch2008}.

Similarly, there has been a plethora of research on estimating win probabilities for March Madness matchups.
Publicly available NCAA basketball team ratings have been around for decades, for instance from \citet{KenMassey}, \citet{KenPomeroy}, \citet{Sagarin}, \citet{SonnyMoore}, and \citet{LRMC}.
Other early approaches from \citet{moreProbModels} and \citet{kvam06} use simple logistic regression models to rate teams.
Since then, modelers have aggregated existing team rating systems into ensemble models.
For instance, \citet{carlin2005improved} uses a prediction model that merges Vegas point spreads with other team strength models, \citet{LopezMatthews} merge point spreads with possession based team efficiency metrics, and \citet{538_ncaa_elo} combines some of these publicly available ratings systems with their own ELO ratings.
Today, people use machine learning or other more elaborate modeling techniques to build team ratings systems.
For instance, ESPN's BPI uses a Bayesian hierarchical model to predict each team's projected point differential against an average Division I team on a neutral court \citep{BPI}.
Further, \citet{Ji2015MarchMP} use a matrix completion approach, \citet{predGBM} uses gradient boosting, \citet{neuralnetpreds} uses a neural network, \citet{randomforestpreds} use a random forest, and \citet{7176206}, \citet{MLcomp}, and \citet{kimwoomagnusen} compare various machine learning models.
Finally, each year there is a popular Kaggle competition in which contestants submit win probabilities for each game and are evaluated on the log-loss of their probabilities \citep{march-machine-learning-mania-2023}.

\subsection{Previous approaches to submitting multiple brackets to a March Madness bracket challenge}\label{app:prevWorkMMmultipleBrackets}

For March Madness bracket challenges, there has been limited research on what we should do with team ratings and win probability estimates once we obtain them.
Most existing research has focused on filling out one optimal bracket after obtaining win probabilities.
For instance, \citet{10.2307/2661505} find the bracket which maximizes expected score and \citet{RePEc:inm:oropre:v:55:y:2007:i:6:p:1163-1177} find the bracket which maximizes expected return conditional on the behavior of other entrants' submitted brackets. 
There has also been some work on filling out multiple brackets in the sports analytics community.
For instance, Scott Powers and Eli Shayer created an \texttt{R} package \texttt{mRchmadness}\footnote{
    \url{https://github.com/elishayer/mRchmadness}
} which uses simulation methods to generate an optimal set of brackets.
Also, Tauhid Zaman at the 2019 Sports Analytics conference\footnote{
    \url{https://www.youtube.com/watch?v=mAgb8A2GDAQ}
} used integer programming to greedily generate a sequence of ``optimal'' brackets subject to ``diversity'' constraints, which force the next bracket in the sequence to be meaningfully different from prior brackets.
Nonetheless, we are not aware of any papers which focus on filling out multiple brackets so as to optimize maximum score or win probability.

\section{Guessing a randomly drawn bitstring details}\label{app:smm_details}

\subsection{Expected maximum score}\label{app:smm_eMaxScore_details}

The expected maximum score of $n$ submitted brackets is
\begin{align}
     & \E \bigg[ \max_{j=1,...,n} f(x^{(j)},\tau) \bigg] \\
    =& \sum_{a=0}^{m} \P \bigg( \max_{j=1,...,n} f(x^{(j)},\tau) > a \bigg) \quad \text{by tail sum} \\
    =& \sum_{a=0}^{m} \bigg\{ 1 - \P \bigg( \max_{j=1,...,n} f(x^{(j)},\tau) \leq a\bigg)\bigg\} \\
    =& \sum_{a=0}^{m} \bigg\{ 1 - \sum_{u=0}^{m} \P \bigg( \max_{j=1,...,n} f(x^{(j)},\tau) \leq a \bigg| u \bigg) \P(u) \bigg\},
\end{align}
where $u = (u_1,...,u_R)$ and $u_\rd$ is the number of zeros in $\tau$ in round $\rd$.
With this definition of $u$, $\{ f(x^{(j)},\tau) \}_{j=1}^{n}$ are conditionally i.i.d. given $u$ and 
\begin{equation}
\P(u) = \prod_{\rd=1}^R \P(u_\rd) = \prod_{\rd=1}^R \dbinom(u_\rd,m_\rd,1-p_\rd).
\end{equation}
Thus,
\begin{align}
    & \E \bigg[ \max_{j=1,...,n} f(x^{(j)},\tau) \bigg] \\
    =& \sum_{a=0}^{m} \bigg\{ 1 - \sum_{u=0}^{m} \P \bigg(f(x^{(j)},\tau) \leq a \text{ for all $j$} \bigg| u \bigg) \P(u) \bigg\} \\
    =& \sum_{a=0}^{m} \bigg\{ 1 - \sum_{u=0}^{m} \P \bigg(f(x^{(1)},\tau) \leq a \bigg| u \bigg)^n \P(u) \bigg\}.
\end{align}
The CDF of the score given $u$ is 
\begin{align}
     & \P \bigg(f(x^{(1)},\tau) \leq a \bigg| u \bigg) \\
    =& \P\bigg(\sum_{\rd=1}^{R} \sum_{i=1}^{m_{\rd}} w_{\rd,i} \cdot \one\{x^{(1)}_{\rd,i} = \tau_{\rd,i}\} \leq a \bigg| u  \bigg) \\
    =& \P\bigg( \sum_{\rd=1}^{R} w_\rd \cdot \big( \Binom(u_\rd, 1-q_\rd) + \Binom(m_\rd - u_\rd, q_\rd) \big) \leq a \bigg). 
\end{align}
This is the CDF of a generalized Poisson Binomial distribution, which we compute in $\Rstu$ using the \textsf{PoissonBinomial} package \citep{PoissonBinomial}.

\subsection{Win probability}\label{app:smm_wp_details}

The probability that the maximum score of our $n$ submitted brackets exceeds or ties that of $k$ opposing brackets is
{\allowdisplaybreaks
\begin{align}
    & \P\bigg[ \max_{j=1,...,n} f(x^{(j)},\tau) \geq \max_{\ell=1,...,k} f(y^{(k)},\tau) \bigg] \\
    =& 1 - \P \bigg[ \max_{j=1,...,n} f(x^{(j)},\tau) < \max_{\ell=1,...,k} f(y^{(k)},\tau)  \bigg] \\
    =& 1 - \P \bigg[  f(x^{(j)},\tau) < \max_{\ell=1,...,k} f(y^{(k)},\tau) \ \forall j \bigg] \\
    =& 1 - \sum_{u=0}^{m} \P \bigg[ f(x^{(j)},\tau) < \max_{\ell=1,...,k} f(y^{(k)},\tau) \ \forall j \bigg| u \bigg] \P(u) \\
    =& 1 - \sum_{u,a=0}^{m} \P \bigg[  f(x^{(j)},\tau) < a \ \forall j \bigg| u \bigg] \P\bigg( \max_{\ell=1,...,k} f(y^{(k)},\tau) = a \bigg| u\bigg) \P(u) \\
    =& 1 - \sum_{u,a=0}^{m} \P \bigg[ f(x^{(1)},\tau) < a \bigg| u \bigg]^n \bigg\{ \P\bigg( \max_{\ell=1,...,k} f(y^{(k)},\tau) \leq a \bigg| u\bigg) - \P\bigg( \max_{\ell=1,...,k} f(y^{(k)},\tau) \leq a -1 \bigg| u\bigg) \bigg\} \P(u) \\
    =& 1 - \sum_{u,a=0}^{m} \P \bigg[ f(x^{(1)},\tau) \leq a-1 \bigg| u \bigg]^n \bigg\{ \P\bigg( f(y^{(1)},\tau) \leq a \bigg| u\bigg)^k - \P\bigg( f(y^{(1)},\tau) \leq a-1 \bigg| u\bigg)^k \bigg\} \P(u).
\end{align}
}

Here, we condition on $u = (u_1,...,u_R)$ where $u_\rd$ is the number of zeros in $\tau$ in round $\rd$.
With this definition of $u$, both $\{ f(x^{(j)},\tau) \}_{j=1}^{n}$ and $\{ f(y^{(\ell)},\tau) \}_{\ell=1}^{k}$ are conditionally i.i.d. given $u$,
We compute the Generalized Poisson Binomial CDFs of the scores $f(x^{(1)},\tau)$ and $f(y^{(1)},\tau)$ given $u$ as described in Appendix~\ref{app:smm_eMaxScore_details}.

In Figure~\ref{fig:wp_smm_varyK} we visualize this win probability as a function of $q$ and $r$ for $p = 0.75$ and various values of $k$ and $n$.

\begin{figure}[p]
  \centering
  \subfloat[]{\includegraphics[width=.42\textwidth]{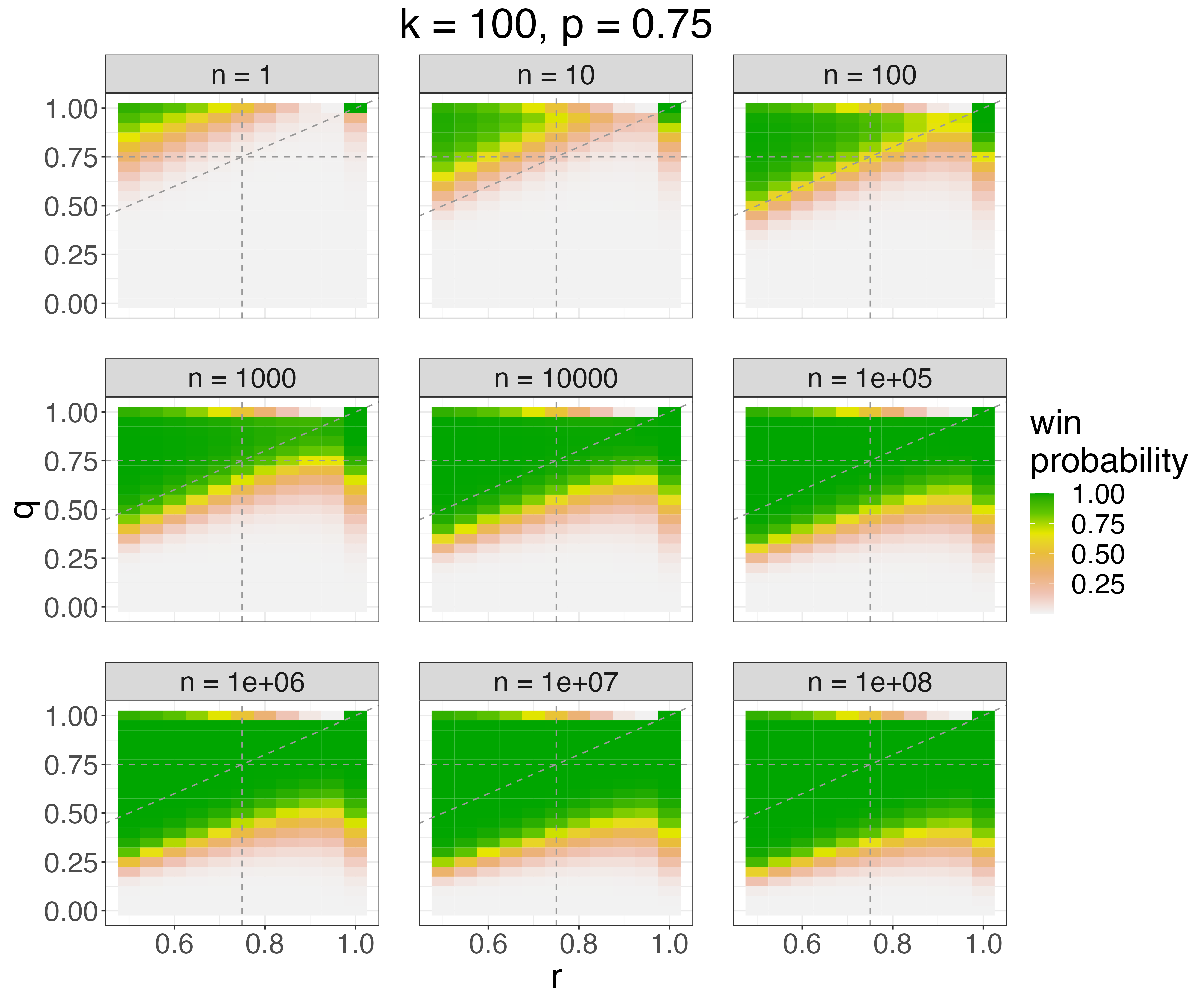}}\quad
  \subfloat[]{\includegraphics[width=.42\textwidth]{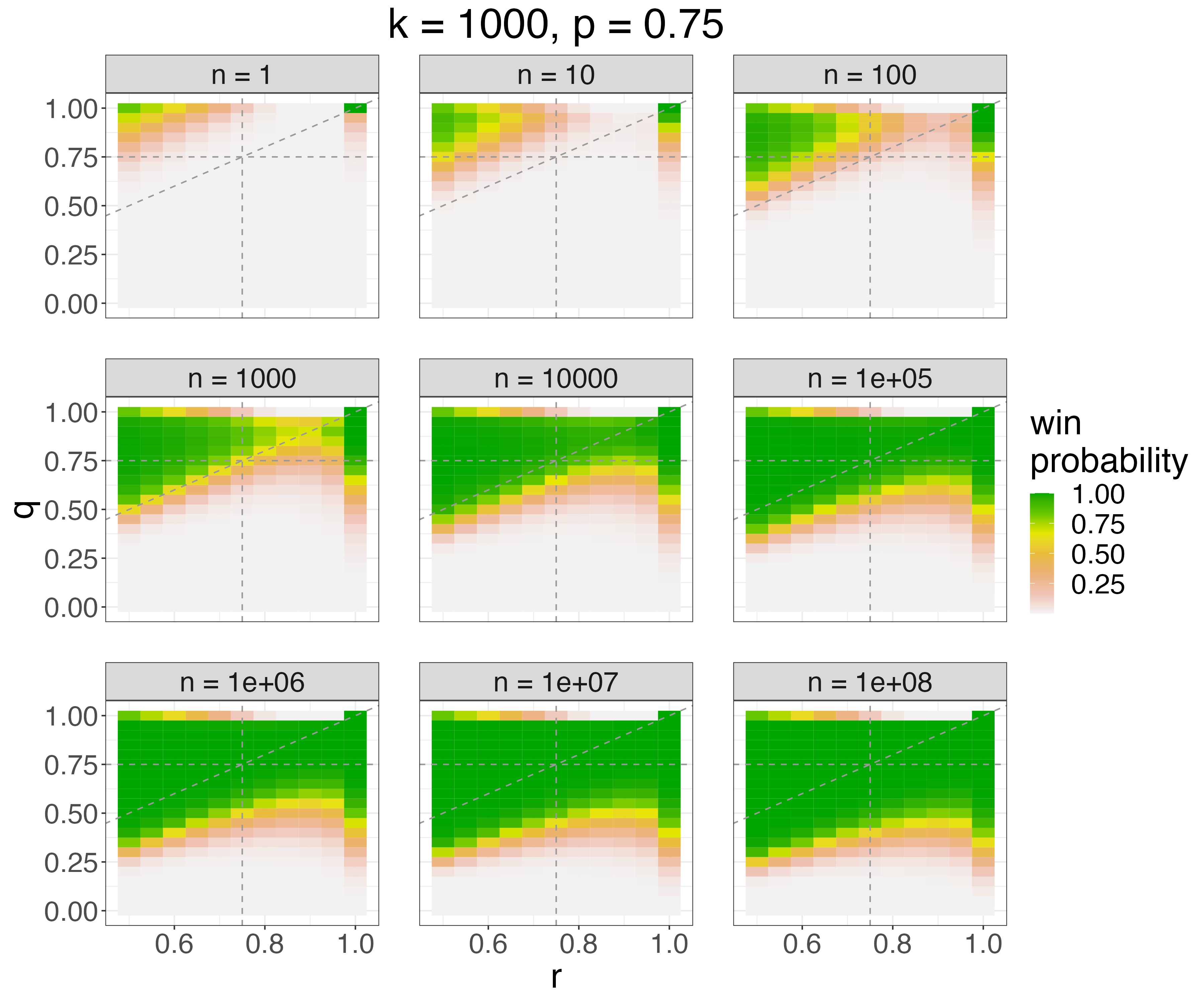}} \\
  \subfloat[]{\includegraphics[width=.42\textwidth]{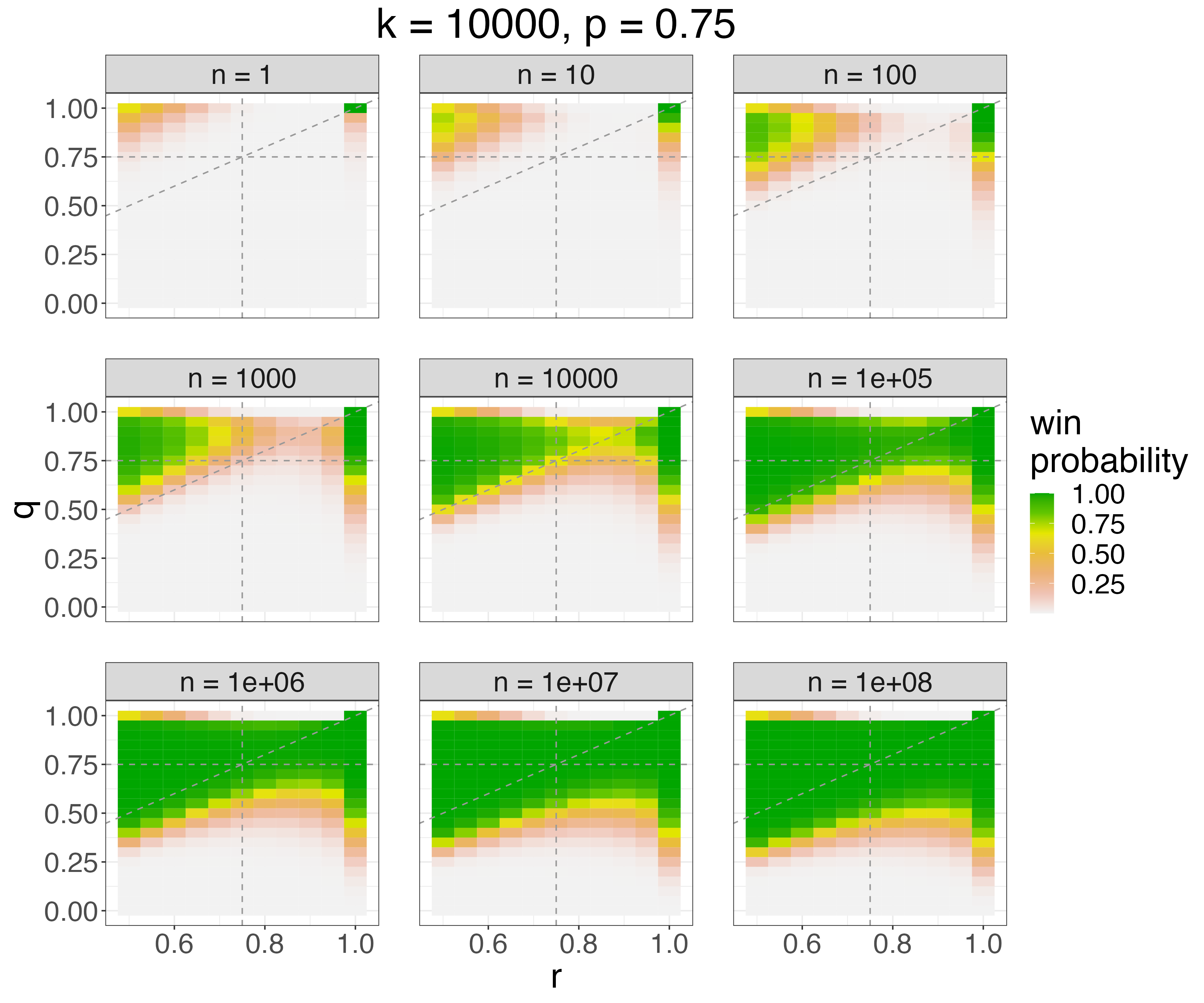}}\quad
  \subfloat[]{\includegraphics[width=.42\textwidth]{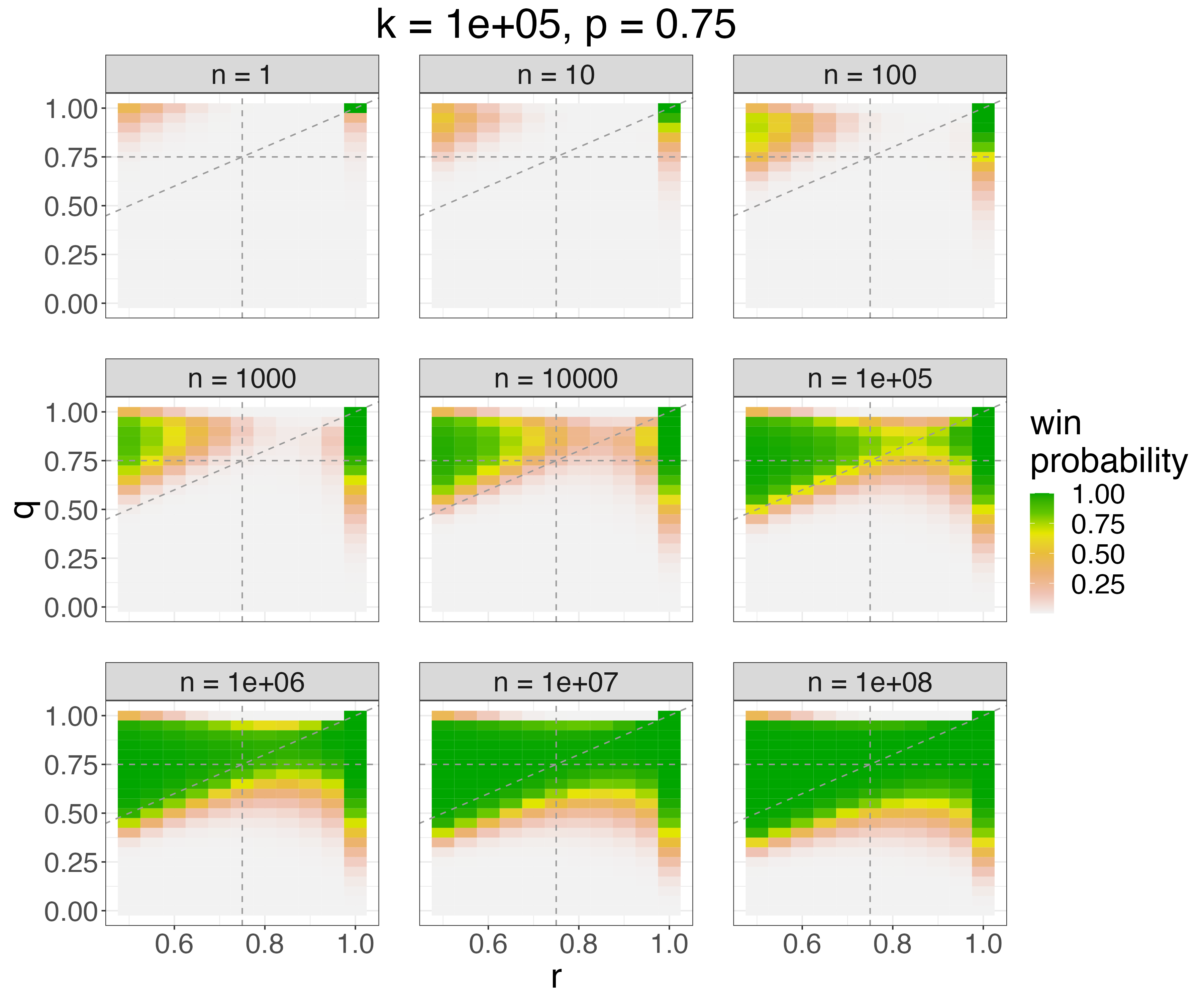}}\\
  \subfloat[]{\includegraphics[width=.42\textwidth]{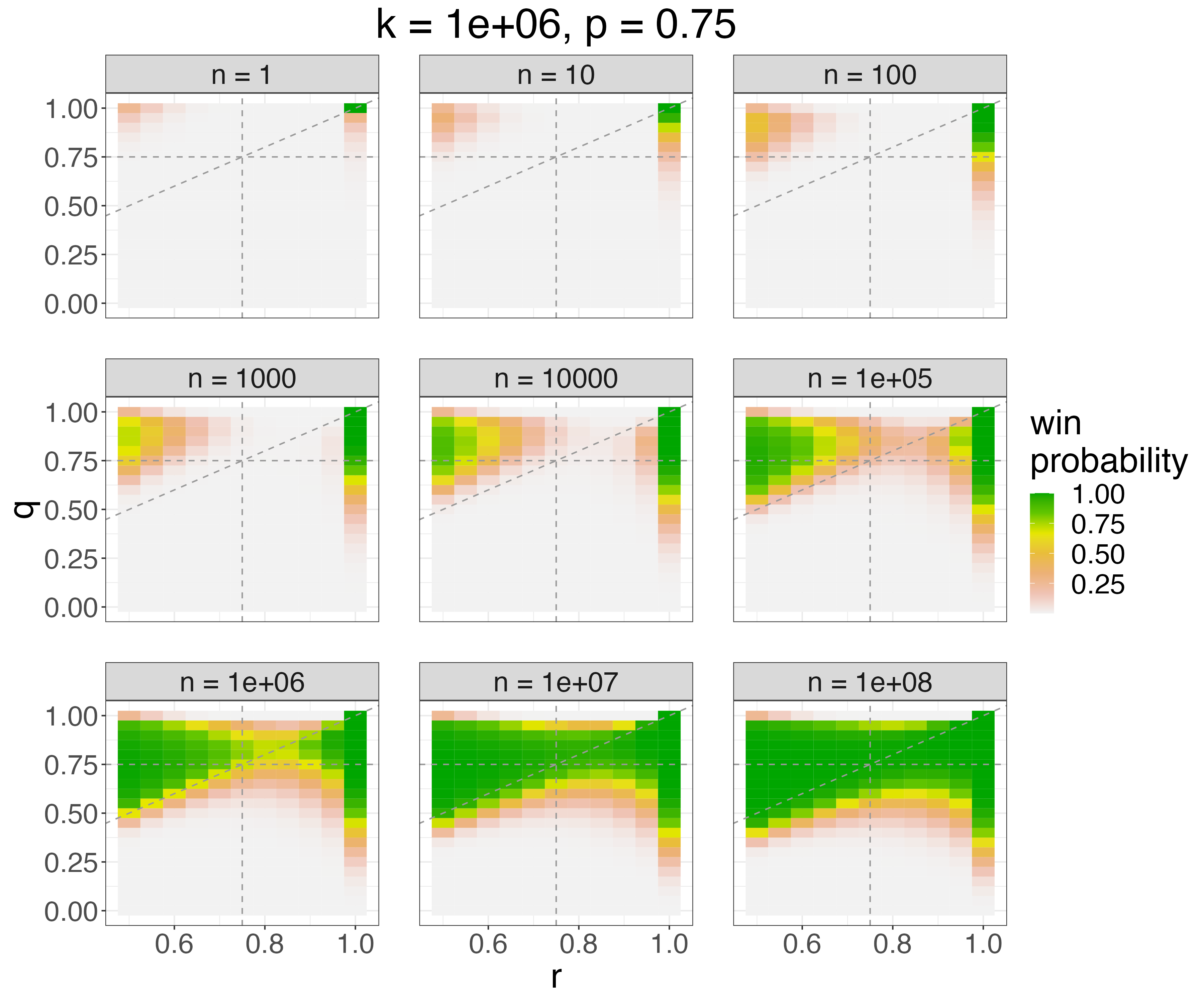}}\quad
  \subfloat[]{\includegraphics[width=.42\textwidth]{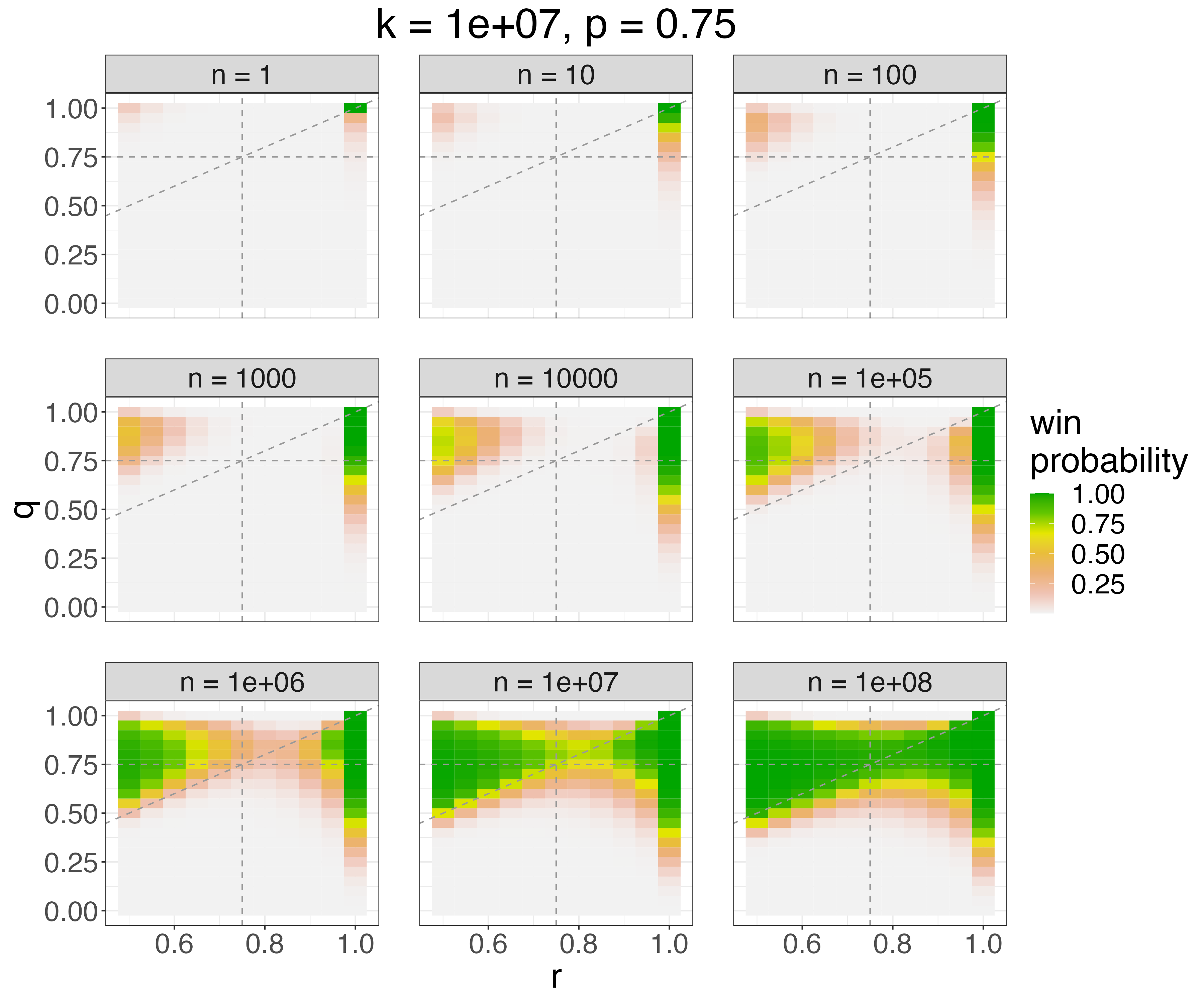}}\\
  \caption{
    The probability (color) that the maximum Hamming score of $n$ submitted Bernoulli($q$) brackets relative to a reference Bernoulli($p$) bracket exceeds that of $k$ opposing Bernoulli($r$) brackets as a function of $q$ ($y$-axis), $r$ ($x$-axis), $n$ (facet), and $k$ (letter) for $p=0.75$ in the ``guessing a randomly drawn bitstring'' contest with $p \equiv p_\rd$, $q \equiv q_\rd$, $r \equiv r_\rd$, and $R=6$ rounds.
  }
  \label{fig:wp_smm_varyK}
\end{figure}

\section{The Asymptotic Equipartition Property}\label{app:EP_thm}

Let $\X$ denote the set of all brackets of length $m$.
Each bracket $x \in \X$ consists of $m$ individual forecasts $x = (x_1,...,x_m)$.
Each forecast $x_i$ has $\o \geq 2$ possible outcomes.
Let $\P$ be a probability measure on $\X$.
The \textit{entropy} of $(\X, \P)$ is $H := \E[-\frac{1}{m}\log_2\P(X)]$, where $X$ is a bracket randomly drawn from $\X$ according to $\P$,
and the \textit{entropy} of a bracket $x \in \X$ is $H(x) := -\frac{1}{m}\log_2\P(x)$.
For instance, in our ``guessing a bitstring'' example from Section~\ref{sec:guessBitstring}, $\X$ is the set of all bitstrings of length $m$, each individual forecast is a bit, and supposing a bitstring is randomly drawn by $m$ independent Bernoulli($p$) coin flips, the entropy  is $H = -(p\log_2(p) + (1-p)\log_2(1-p))$.

Letting $\epsilon > 0$, we partition the set of all brackets $\X$ into three subsets,
\begin{equation}
    \begin{cases}
    \text{$\epsilon$-low entropy ``chalky'' brackets} \qquad \C_\epsilon := \{x \in \X : \P(x) \geq 2^{-m(H-\epsilon)} \}, \\
    \text{$\epsilon$-``typical'' brackets} \qquad \T_\epsilon := \{x \in \X : 2^{-m(H+\epsilon)} < \P(x) < 2^{-m(H-\epsilon)} \}, \\
    \text{$\epsilon$-high entropy ``rare'' brackets} \qquad \quad \R_\epsilon := \{x \in \X : \P(x) \leq 2^{-m(H+\epsilon)} \}.
    \end{cases}
\end{equation}
Under this definition, an individual chalky bracket is more probable than an individual typical bracket, which is more probable than an individual rare bracket.

The \textit{Asymptotic Equipartition Property} (A.E.P.) from Information Theory, Theorem~\ref{thm:equipartition_thm}, quantifies our intuition about chalky, typical, and rare brackets
from Section~\ref{sec:ep_sec}:
as $m$ tends to infinity, the probability mass of the set of brackets becomes increasingly more concentrated in an exponentially small set, the typical set.
The proof of Theorem \ref{thm:equipartition_thm} is adapted from \citet{cover_thomas_info_thy}.

The primary mathematical takeaway from Theorem~\ref{thm:equipartition_thm} is as follows.
The set $\X$ of all possible length $m$ brackets has exponential size $\o^m$, recalling that $\o$ is the number of possible outcomes of each individual forecast (e.g., $\o=2$ in ``guessing a bitstring'').
The $\epsilon$-typical set $\T_\epsilon$ comprises a tiny fraction of $\X$, having size $|\T_\epsilon| \approx 2^{mH}$ by part (b) of the theorem.
Therefore, $\T_\epsilon$ is exponentially smaller than $\X$,
$|\T_\epsilon|/|\X| \approx 2^{mH}/\o^m = 2^{-m(\log_2\o - H)}$.
In the ``guessing a bitstring'' example with $p = 0.75$ in which the reference bitstring consists of $m$ independent Bernoulli($p$) bits, $\o = 2$ and $H \approx 0.81$.
Thus, $|\T_\epsilon|/|\X| \approx 2^{-m(0.19)} \approx 0.88^m$.
When $m=63$ (as in March Madness), $|\T_\epsilon|/|\X| \approx 0.00026$, so the typical set of brackets is about $4,000$ times as small as the full set.
This factor increases exponentially as $m$ increases.

\begin{theorem}[Asymptotic Equipartition Property]
\label{thm:equipartition_thm}
\phantom\\
    Let $\epsilon > 0$.
\begin{enumerate}
    \item[(a)] The typical set asymptotically contains most of the probability mass:
    \begin{equation}
        \P(\T_\epsilon) \to 1 \text{ in probability as } m \to \infty.
    \end{equation}
    In other words, the reference bitstring is likely a typical bracket. 
    \item[(b)] For $m$ sufficiently large, we can bound the sizes of sets of chalky, typical, and rare brackets in terms of the entropy,
    \begin{equation}
    \begin{cases}
        |\C_\epsilon| < 2^{m(H-\epsilon)}, \\
        \big(1 - \epsilon) \cdot 2^{m(H-\epsilon)} < |\T_\epsilon| < 2^{m(H+\epsilon)}, \\
        |\R| > \o^m - 2^{m(H+\epsilon)} - 2^{m(H-\epsilon)}.
    \end{cases}
    \end{equation}
    In other words, most brackets are rare, exponentially fewer brackets are typical, and exponentially fewer of those are chalky.
    \item[(c)] For $m$ sufficiently large, the typical set is essentially the smallest high probability set:
    letting $\delta > 0$ and $B_\delta \subset \X$ be any high probability set with $\P(B_{\delta}) \geq 1 - \delta$, $B_\delta$ and $\T_\epsilon$ have similar sizes, $|B_{\delta}| \geq \big(1 - \epsilon - \delta\big)2^{m(H-\epsilon)}$.
    $B_\delta$ is a high probability set when $\delta$ is small, and in that case both $|B_\delta|$ and $|T|$ are essentially bounded below by $\big(1 - \epsilon)2^{m(H-\epsilon)}$.
    \end{enumerate}
\end{theorem}


\begin{proof}[Proof (Theorem \ref{thm:equipartition_thm})]

\begin{align}
    \P(\T_\epsilon) = \P_{X \sim \P}(X \in \T_\epsilon) = \P\big( 2^{-m(H+\epsilon)} < \P(X) < 2^{-m(H-\epsilon)} \big) = \P\big( \big| -\frac{1}{m}\log_2\P(x) - H \big| \geq \epsilon\big),
\end{align}
which converges to 1 by the law of large numbers since $H = \E[-\frac{1}{m}\log_2\P(X)]$. 
This proves part (a).

Now,
\begin{align}
    1 = \sum_{x \in \X} \P(x) \geq \sum_{x \in \T_\epsilon} \P(x) > \sum_{x \in \T_\epsilon} 2^{-m(H+\epsilon)} = 2^{-m(H+\epsilon)} \cdot |\T_\epsilon|,
\end{align}
so $|\T_\epsilon| < 2^{m(H+\epsilon)}$.
By part (a), for $m$ sufficiently large,
\begin{align}
    1 - \epsilon \leq \P(\T_\epsilon) = \sum_{x \in \T_\epsilon} \P(x) < \sum_{x \in \T_\epsilon} 2^{-m(H-\epsilon)} = 2^{-m(H-\epsilon)} \cdot |\T_\epsilon|, 
\end{align}
so $|\T_\epsilon| > (1 - \epsilon)\cdot2^{m(H-\epsilon)}$. 
Similarly,
\begin{align}
    1 = \sum_{x \in \X} \P(x) \geq \sum_{x \in \C_\epsilon} \P(x) > \sum_{x \in \C_\epsilon} 2^{-m(H-\epsilon)} = 2^{-m(H-\epsilon)} \cdot |\C_\epsilon|,
\end{align} 
so $|\C_\epsilon| < 2^{m(H-\epsilon)}.$
Therefore,
\begin{align}
|\R_\epsilon| &= |\X \setminus (\T_\epsilon \cup \C_\epsilon) | = \o^m - |\T_\epsilon| - |\C_\epsilon| > \o^m - 2^{m(H+\epsilon)} - 2^{m(H-\epsilon)}.
\end{align}
This proves part (b).

Finally, by part (a), for $m$ sufficiently large,
\begin{align}
1 - \delta - \epsilon = (1 - \epsilon) + (1-\delta) - 1 \leq \P(\T_\epsilon) + \P(B_\delta) - \P(\T_\epsilon \cup B_\delta).
\end{align}
Thus,
\begin{align}
\begin{split}
1 - \delta - \epsilon 
&\leq \P(\T_\epsilon \cap B_\delta) 
= \sum_{x \in \T_\epsilon \cap B_\delta} \P(x) 
\leq \sum_{x \in \T_\epsilon \cap B_\delta}  2^{-m(H - \epsilon)} \\
&= |\T_\epsilon \cap B_\delta|  \cdot 2^{-m(H - \epsilon)} 
\leq |B_\delta| \cdot 2^{-m(H - \epsilon)},
\end{split}
\end{align}
so $|B_{\delta}| \geq \big(1 - \epsilon - \delta\big)2^{m(H-\epsilon)}$.
This proves part (c).
\end{proof}

\section{Pick six details}\label{app:pick_six_details}

We can explicitly and quickly compute a tractable lower bound for the expected profit (Formula~\eqref{eqn:pick_six_Eprofit}) under our pick six model from Section~\ref{sec:ex_pickSix}.
We begin with 
\begin{align}
   \E\bigg( \frac{W}{W + \Wopp} \bigg) 
   &= \E_{\tau\sim\bP, x\sim\bQ,y\sim\bR}\bigg( \frac{W}{W + \Wopp} \bigg) \\
   &= \sum_{\tau}\P(\tau) \E\bigg( \frac{W}{W + \Wopp} \bigg|\tau \bigg) \\
   &= \sum_{\tau}\P(\tau) \sum_{w,w'} \big( \frac{w}{w + w'} \big) \P(W=w,\Wopp=w'|\tau) \\
   &= \sum_{\tau}\P(\tau) \sum_{w,w'} \big( \frac{w}{w + w'} \big) \P(W=w|\tau)\P(\Wopp=w'|\tau) 
\end{align}
since $W$ is conditionally independent of $\Wopp$ given $\tau$,
\begin{align}
   &= \sum_{\tau}\P(\tau) \sum_{w'} \sum_{w\geq1} \big( \frac{w}{w + w'} \big) \P(W=w|\tau)\P(\Wopp=w'|\tau) 
\end{align}
since if $w=0$, $w/(w+w')=0$,
\begin{align}
   &\geq \sum_{\tau}\P(\tau) \sum_{w'} \sum_{w\geq1} \big( \frac{1}{1 + w'} \big) \P(W=w|\tau)\P(\Wopp=w'|\tau) 
\end{align}
since $w/(w+w')\geq 1/(1+w')$, which is essentially to say that we won't submit duplicate tickets,
\begin{align}
   &= \sum_{\tau}\P(\tau) \P(W\neq0|\tau) \sum_{w'} \big( \frac{1}{1 + w'} \big) \P(\Wopp=w'|\tau) \\
   &= \sum_{\tau}\P(\tau) \P(W\neq0|\tau) \E\bigg[\frac{1}{1+\Wopp} \bigg| \tau\bigg] \\
   &\geq \sum_{\tau}\P(\tau) \P(W\neq0|\tau) \frac{1}{1+\E[\Wopp|\tau]}  
\end{align}
by Jensen's inequality, since $x \mapsto 1/(1+x)$ is convex when $x>0$.

Now,
\begin{align}
    \P(\tau) = \P_{\tau\sim\bP}(\tau) = \P(\tau_1,...,\tau_s) = \prod_{j=1}^{s}\P(\tau_j) = \prod_{j=1}^{s}\bP_{\tau_j j}.
\end{align}
Also,
\begin{align}
    \P(W\neq0'|\tau) 
    &= \P_{\tau\sim\bP,x\sim\bQ}(W\neq0'|\tau) \\
    &= \P(\exists \ell\in\{1,...,n\} \text{ such that } x^{(\ell)}=\tau|\tau) \\
    &= 1 - \P(\forall \ell\in\{1,...,n\}, \ x^{(\ell)} \neq \tau|\tau) \\
    &= 1 - \P(x^{(1)} \neq \tau|\tau)^n 
\end{align}
since the $\{x^{(\ell)}\}$ are i.i.d.,
\begin{align}
    &= 1 - \P(\exists j\in\{1,...,s\} \text{ such that }  \ x^{(1)}_j \neq \tau_j|\tau)^n \\
    &= 1 - \big(1 - \P(\forall j\in\{1,...,s\}, \  \ x^{(1)}_j = \tau_j|\tau)\big)^n  \\
    &= 1 - \big(1 - \prod_{j=1}^{s} \P(x^{(1)}_j = \tau_j|\tau)\big)^n  
\end{align}
since each of the $s$ races are independent,
\begin{align}
&= 1 - \big(1 - \prod_{j=1}^{s} \bQ_{\tau_j j} \big)^n.
\end{align}
Then, by similar logic,
\begin{align}
    \E[\Wopp|\tau]
    &= \E_{\tau\sim\bP,y\sim\bR}[\Wopp|\tau] \\
    &= \E\bigg[ \sum_{\ell=1}^{k} \one\{y^{(\ell)} = \tau\} \bigg|\tau \bigg] \\
    &= \sum_{\ell=1}^{k} \P(y^{(\ell)} = \tau |\tau) \\
    &= k\cdot \P(y^{(1)} = \tau |\tau) \\
    &= k\cdot \prod_{j=1}^{s} \P(y^{(1)}_j = \tau_j |\tau) \\
    &= k\cdot \prod_{j=1}^{s} \bR_{\tau_j j}.
\end{align}
Combining all these formulas, we can explicitly and quickly evaluate a lower bound for the expected profit,
\begin{align}
    \E[\text{Profit}] &= -n + T\cdot\E\bigg( \frac{W}{W + \Wopp} \bigg) \\
    &\geq -n + T \cdot \sum_{\tau}\P(\tau) \P(W\neq0|\tau) \frac{1}{1+\E[\Wopp|\tau]} \\
    &= -n + T \cdot \sum_{\tau} 
    \bigg(\prod_{j=1}^{s}\bP_{\tau_j j}\bigg)
    \bigg( 1 - \big(1 - \prod_{j=1}^{s} \bQ_{\tau_j j} \big)^n \bigg)
    \bigg( \frac{1}{1+ k\cdot \prod_{j=1}^{s} \bR_{\tau_j j}} \bigg).
\label{eqn:pick_six_Eprofit_full}
\end{align}


\end{document}